\documentclass[11pt]{article}%
\usepackage{amsmath}
\usepackage{amssymb}
\usepackage{amsfonts}

\usepackage{cite}
\usepackage{graphicx}
\usepackage{float}
\usepackage{longtable}
\usepackage[utf8]{inputenc}
\usepackage{hyperref}

\usepackage{algorithmicx}
\usepackage[ruled,vlined]{algorithm2e}
\SetAlgoCaptionSeparator{ }

\usepackage{fullpage}
\usepackage{authblk}

\newcommand{\fref}[1]{Fig.~\ref{#1}}

\newcommand{\tref}[1]{Table~\ref{#1}}
\newcommand{\sref}[1]{Section~\ref{#1}}

\newenvironment{algo}[1][!h]
  {
   \begin{algorithm}[#1]%
  }{\end{algorithm}}

\newenvironment{proced}[1][!h]
  {
   \begin{algorithm}[#1]%
  }{\end{algorithm}}

\providecommand{\U}[1]{\protect\rule{.1in}{.1in}}
\newtheorem{theorem}{Theorem}

\newtheorem{lemma}{Lemma}

\newtheorem{proposition}{Proposition}

\newenvironment{proof}[1][Proof]{\textbf{#1.} }{\ \rule{0.5em}{0.5em}}
\begin{document}

\title{Multilayer Optimization for the Quantum Internet}
\author[1,2,3]{Laszlo Gyongyosi\footnote{Email: \href{mailto:l.gyongyosi@soton.ac.uk}{l.gyongyosi@soton.ac.uk}}}
\author[2]{Sandor Imre}
\affil[1]{School of Electronics and Computer Science, University of Southampton, Southampton, SO17 1BJ, UK}
\affil[2]{Department of Networked Systems and Services, Budapest University of Technology and Economics, Budapest, H-1117 Hungary}
\affil[3]{MTA-BME Information Systems Research Group, Hungarian Academy of Sciences, Budapest, H-1051 Hungary}
\date{}

\maketitle

\begin{abstract}
We define a multilayer optimization method for the quantum Internet. Multilayer optimization integrates separate procedures for the optimization of the quantum layer and the classical layer of the quantum Internet. The multilayer optimization procedure defines advanced techniques for the optimization of the layers. The optimization of the quantum layer covers the minimization of total usage time of quantum memories in the quantum nodes, the maximization of the entanglement throughput over the entangled links, and the reduction of the number of entangled links between the arbitrary source and target quantum nodes. The objective of the optimization of the classical layer is the cost minimization of any auxiliary classical communications. The multilayer optimization framework provides a practically implementable tool for quantum network communications, or long-distance quantum communications.
\end{abstract}

\section{Introduction}
\label{sec1}
Quantum Internet is a communication network with quantum nodes and quantum links \cite{ref1, ref2, ref3, ref4, ref5, ref6, ref7, ref31, ref42,ref43,ref44,ref45,ref46} that allows to the parties to perform efficient quantum communications \cite{ref34,ref35,ref36}. The aim of quantum Internet \cite{ref5, ref31} and quantum repeater networks \cite{ref1,ref34,ref35,ref36,ref37,ref38,ref39,ref40} is to distribute quantum entanglement between distant nodes through a chain of intermediate quantum repeater nodes \cite{ref21, ref22, ref23, ref24, ref25, ref26, ref27, ref28, ref29, ref30}. In the quantum Internet, the quantum nodes share entangled connections that formulate entangled links \cite{ref1, ref2, ref3, ref4, ref5, ref6}. The quantum nodes store the quantum states in their local quantum memory for path selection and path recovery purposes \cite{ref1, ref2, ref3, ref21, ref22, ref23, ref24, ref25, ref26, ref27, ref28, ref29, ref30, ref31}. Since setting several attributes must be optimized in parallel in an arbitrary quantum network, the optimization problem formulates a multi-objective procedure. Formally, multi-objective optimization covers the minimization of quantum memory usage time (storage time), the maximization of entanglement throughput (number of transmitted entangled states per second of a particular fidelity) of entangled links \cite{ref1, ref2, ref3, ref4, ref5, ref6}, and the reduction of the number of entangled links between a source and a target quantum node \cite{ref21, ref22, ref23, ref24, ref25, ref26, ref27, ref28, ref29, ref30, ref31}. However, the problem does not end here since a quantum repeater network can be approached on the quantum transmission (quantum layer) level and on the auxiliary classical communication (classical layer) level that is required for the dynamic functioning of the quantum layer. Therefore, the problem is not just a multi-objective optimization problem in the quantum layer but also a multilayer optimization issue that covers the development of both the quantum and classical layers of a quantum repeater network.    

Here, we define a multilayer optimization method for quantum repeater networks. This covers both the quantum layer and the classical layer of a quantum repeater network. 

By utilizing the tools of quantum Shannon theory \cite{ref1, ref2, ref3, ref4, ref5, ref6, ref7}, the optimization of the quantum layer includes minimizing the usage of quantum memories in the nodes to reduce the storage time of entangled states, the maximization of entanglement throughput of the entangled links, and also these conditions have to be satisfied for the shortest path between a given source node and target quantum node (i.e., a multi-objective optimization of the quantum layer).

The aim of classical-layer optimization is to curtail the cost of auxiliary classical communications, which is required for such optimization. The cost of the classical communication covers all communication costs required to achieve the optimal quantum network state including the classical communication steps for overall quantum storage time minimization, entanglement throughput maximization, and the selection of a shortest path. 

The multilayer optimization employs advanced methods to solve the multi-objective optimization of the quantum layer. We define the structures of the quantum memory utilization graph and the entanglement throughput tree for the multi-objective optimization of the quantum layer of a quantum repeater network. The quantum memory utilization graph models the quantum memory usage for entanglement storage. The entanglement throughput tree shows the entanglement throughput of entangled links with respect to the number of transmittable entangled states at a particular fidelity. Using these advanced constructions, we also define a method for the optimal assignment of entangled states in the repeater nodes. The input of the quantum layer optimization procedure is the quantum memory utilization graph, while the output of the method is a set of entanglement throughput trees. The output identifies the optimal states of the quantum network with respect to the multi-objective optimization function. 

Classical-layer optimization focuses on the minimization of the total cost of classical communications by utilizing swarm intelligence\cite{ref13, ref14, ref15, ref16, ref17}. This also defines a multi-objective problem since the cost has to be reduced with respect to the classical communication cost required for the minimization of quantum memory usage, the classical cost of entanglement throughput maximization of entangled links, and for the selection of the shortest path in the quantum layer. Classical-layer optimization uses some fundamentals of bacteria foraging models \cite{ref11, ref12, ref13, ref17, ref18, ref19, ref20} and probabilistic multi-objective uncertainty characterization \cite{ref8, ref9,ref10,ref11, ref12, ref13}. 

The optimization framework requires no changes in the physical layer, so the framework is directly implementable by the current physical devices \cite{ref1, ref2, ref3, ref4,ref5, ref6, ref34,ref35,ref36} and quantum networking elements \cite{ref21, ref22, ref23, ref24, ref25, ref26, ref27, ref28, ref29, ref30, ref31}. The method is useful in quantum networking environments with diverse physical attributes (different quantum memory characteristics, quantum error correction, physical quantum nodes attributes, and transmission capabilities of noisy quantum links). 

The novel contributions of our paper are as follows:
\begin{itemize}
\item  \textit{We conceive a complex optimization framework for quantum networks. It integrates the development of the quantum and classical layers of quantum repeater networks.}

\item \textit{Quantum-layer optimization utilizes the attributes of physical-layer quantum transmissions, quantum memory usage, and entanglement distribution via the framework of quantum Shannon theory.}

\item \textit{Classical-layer optimization focuses on minimizing any auxiliary communications related to the quantum layer and optimization.}

\item \textit{Multilayer optimization is applicable by current physical devices and quantum networking elements providing a solution for the optimization of arbitrary quantum networking scenarios with diverse physical attributes and environments.}
\end{itemize}

This paper is organized as follows. In \sref{sec2}, the system model is proposed. In \sref{sec3}, the optimization procedure of the quantum layer of quantum repeater networks is defined. \sref{sec4} studies the optimization of the classical layer. \sref{sec5} provides a performance evaluation. Finally, \sref{sec6} concludes the paper. Supplemental material is included in the Appendix.  
 
\section{System Model}
\label{sec2}
Our model assumes that the quantum repeater network consists of a source and target node with intermediate repeater nodes and a quantum switcher node. A quantum switcher node $S$ operates as follows. Node $S$ is a quantum repeater node capable of switching between the entangled connections stored in its local quantum memory and a permit of applying entanglement swapping on the selected connections. While an $i$-th quantum repeater node establishes only the entangled connections with the neighbor quantum repeater nodes, a switcher node is equipped with an extended knowledge about the quantum network to select between the entangled links. A general repeater node is not allowed to perform any link selection, since it is assumed in the model that a quantum repeater node has only local knowledge about the network. A switcher node based on its network knowledge can also send entanglement swapping commands to the quantum repeater nodes to define new paths in the network.

Let $N$ be a quantum network, $N=\left(V,{\rm {\mathcal S}}\right)$, where $V$ is a set of nodes, ${\rm {\mathcal S}}$ is a set of entangled links. Without loss of generality, the level ${{\text{L}}_{l}}$ of an entangled link $E\left(x,y\right)$ is defined as follows. For an ${{\text{L}}_{l}}$-level entangled link, the hop distance between quantum nodes $x$ and $y$ is \cite{ref1,ref41}
\begin{equation} \label{eq1} 
d{\left(x,y\right)}_{\text{L}_l}=2^{l-1}, 
\end{equation} 
 with $d{\left(x,y\right)}_{\text{L}_l}-1$ intermediate nodes between the nodes $x$ and $y$. The probability that an $\text{L}_l$-level entangled link $E\left(x,y\right)$ exists between nodes $x,y$ is ${\Pr }_{\text{L}_l}\left(E\left(x,y\right)\right)$, which depends on the actual network.

The $S$ quantum switcher is modeled as a quantum node with the following attributes and permissions:
\begin{itemize}
\item \textit{knowledge about the physical attributes of distant quantum repeater nodes and the entangled connections of $N$ (e.g., entanglement fidelity, quantum memory status, link noise, etc),}
\item \textit{internal quantum memory for the storage of entangled states,}
\item \textit{quantum functionality:} 
\begin{itemize}
\item \textit{permission to set new entangled connections between its local quantum system and a selected quantum node of the quantum network,}
\item \textit{permission to switch between the stored entangled states to construct new paths,}
\end{itemize}
\item \textit{classical functionality:}
\begin{itemize}
\item \textit{permission to command distant quantum nodes of $N$ via classical links (to construct new entangled connections in the network, to perform entanglement swapping between the selected nodes, other).}
\end{itemize}
\end{itemize}

The network model is illustrated in \fref{fig1}. The example network in \fref{fig1}(a) consists of six quantum repeater nodes $R_{i} $, $i=1,\ldots ,6$, and a quantum switcher $S$ that switches between the entangled connections using its local quantum memory. The switcher also can perform entanglement swapping in the network. The switcher node $S$ has knowledge about the physical attributes (e.g., entanglement fidelity, quantum memory status, etc) of quantum repeater nodes to make a decision on a path. The knowledge about the repeater nodes can be transmitted over a classical link to the quantum switcher (classical links are not depicted). 
As it is depicted in \fref{fig1}(b), the switcher node has a permission to set new entangled connections via its local quantum state with a selected quantum node. The $S$ switcher node decided to set a new entangled connection between its local quantum system and repeater node $R_{2}$. A standard quantum repeater is not allowed to perform these operations (except with the direct neighbors in the entanglement distribution phase) without a dedicated command from the switcher.

 \begin{center}
\begin{figure*}[!h]
\begin{center}
\includegraphics[angle = 0,width=1\linewidth]{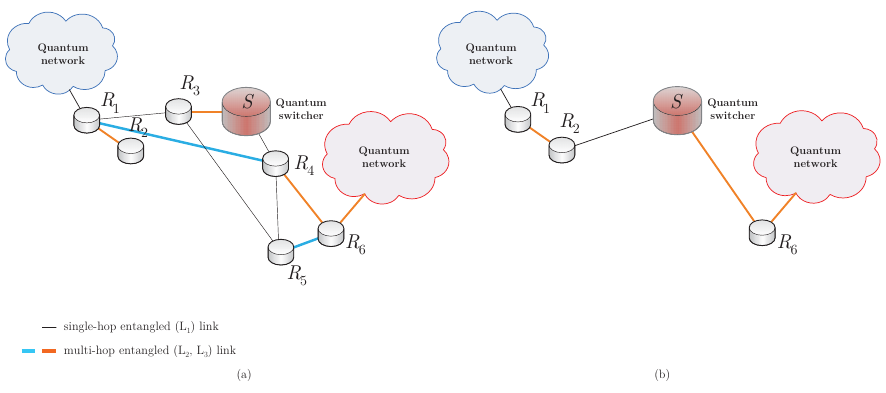}
\caption{The network model with a quantum switcher $S$ and quantum repeater nodes $R_{i} $, $i=1,\ldots ,6$. The ${{\text{L}}_{1}}$, ${{\text{L}}_{2}}$- and ${{\text{L}}_{3}}$-level entangled links of $N$ are depicted by gray, blue and orange, respectively (additional nodes are not shown). (a): For a current shortest path ${\rm {\mathcal P}}_{1} =\left\{R_{1},R_{3},S,R_{4},R_{6} \right\}$, the active repeater nodes selected by $S$ are $R_{3} $ and $R_{4} $. (b): Node $S$ switches the entangled connections in the local quantum memory from $R_{3} $ to $R_{2} $ and from $R_{4} $ to $R_{6} $. The switching operation defines a new shortest path ${\rm {\mathcal P}}_{2} =\left\{R_{1} ,R_{2} ,S,R_{6} \right\}$.} 
 \label{fig1}
 \end{center}
\end{figure*}
\end{center}

Note, path ${\rm {\mathcal P}}_{1} $ in \fref{fig1}(a) provides a shortest path at a particular network situation at an initial network time $T_{1} $.  Since the quantum network $N$ evolves in time (quality of the entangled links, the status of the nodes, internal quantum memories, etc), at a given time $T_{2} $, by utilizing the functions of the switcher node, the switcher node determined a new shortest path, ${\rm {\mathcal P}}_{2} $, as depicted in \fref{fig1}(b). 
 
\subsection{Quantum Memory Scheduling}

In this section, we define a structure for scheduling quantum memory usage of the quantum nodes called the quantum memory utilization graph ${\rm {\mathcal G}}_{m} $. This is a directed graph mapped from the network model, with several abstracted nodes and links.

\begin{proposition}
The ${\rm {\mathcal G}}_{m} $ quantum memory utilization graph is a directed graph with abstract nodes and links to schedule the quantum memory usage mapped from the quantum repeater network.  
\end{proposition}

The ${\rm {\mathcal G}}_{m} $ graph of quantum memory utilization is constructed as follows. Assuming $n$ quantum nodes (excluding the $S$ switcher node) in the network, the graph contains $n$ abstracted transmitter nodes and $n$ abstracted receiver nodes with directed connections. Note if the quantum switcher $S$ is modeled as the $n$-th node and $n_{S}=1$, where $n_{S}$ is the number of quantum switchers in $N$, the ${\rm {\mathcal G}}_{m} $ graph also can be constructed via $\left(n-1\right)$ transmitter and $\left(n-1\right)$ receiver quantum nodes. 
A given $N$ with an arbitrary number of quantum switcher nodes defines a particular ${\rm {\mathcal G}}_{m} $, therefore the ${\rm {\mathcal G}}_{m} $ graph is a combination of all other possible switcher modes. 
The ${\rm {\mathcal G}}_{m} $ graph contains directed edges between quantum node pairs $\left(\left. V_{x} \right|V_{y} \right)$ and $\left(\left. V_{y} \right|V_{z} \right)$ of a particular switcher mode $S_{i} $, $i=1,\ldots ,N_{S} $, where $N_{S} $ is the total number of switcher modes, depicted via nodes labeled as $\left(\left. x\right|y\right)$ and $\left(\left. y\right|z\right)$.

Let us assume that $N$ has a single switcher node $S$, and it has two states, $S_{1} $ and $S_{2} $. In both modes $S_{1} $ and $S_{2} $, repeater node $R_{1} $ serves as a transmitter node for node $R_{2} $. 

In switcher mode $S_{1} $, the shared entangled connection defines the following relations. Repeater node $R_{1} $ serves as a transmitter node for nodes $R_{3} $ and $R_{4} $. Node $R_{3} $ serves as a transmitter node for nodes $R_{4} $ and $R_{5} $. Node $R_{4} $ serves as a transmitter node for nodes $R_{5} $ and $R_{6} $. Node $R_{5} $ serves as a transmitter node for node $R_{6} $.

In switcher mode $S_{2} $, the shared entangled connection defines the following relation. Repeater node $R_{2} $ serves as a transmitter node for node $R_{6} $.

The ${\rm {\mathcal G}}_{m} $ graph of quantum memory utilization derived from the quantum network setting of \fref{fig1}. is illustrated in \fref{fig2}. 

 \begin{center}
\begin{figure*}[!h]
\begin{center}
\includegraphics[angle = 0,width=1\linewidth]{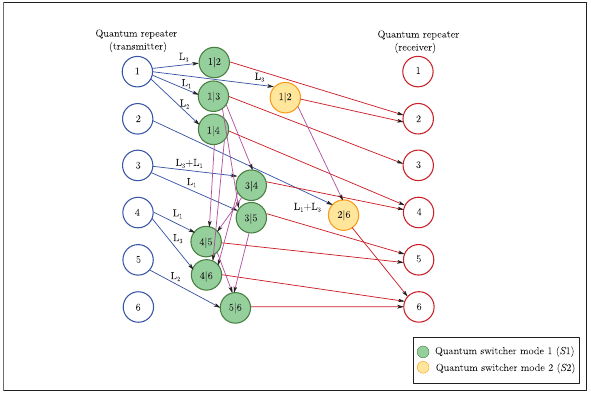}
\caption{The ${\rm {\mathcal G}}_{m} $ graph of quantum memory utilization derived from the network setting in Fig. 1. The graph contains $n$ abstracted transmitter nodes and $n$ abstracted receiver nodes with directed entangled connections. An $i$-level entangled connect is depicted by ${\rm L}_{i} $. The $S$ quantum switcher has two states, $S_{1} $ and $S_{2} $. The green circles represent quantum nodes operating on switcher mode $S_{1} $. The yellow circles represent nodes operating on switcher mode $S_{2} $.} 
 \label{fig2}
 \end{center}
\end{figure*}
\end{center}

\subsection{Entanglement Throughput Tree}
In this section, we define the structure of the entanglement throughput tree, which aims to extract information from the quantum memory utilization graph. The entanglement throughput tree is also the output format of the multi-objective optimization procedure of the quantum layer. 
\begin{lemma}
The ${\rm {\mathcal G}}_{et} $ entanglement throughput tree is a structure modeling the multi-objective optimization problem of quantum repeater networks. The tree structure is derived from the ${\rm {\mathcal G}}_{m} $ quantum memory utilization graph.
\end{lemma}
\begin{proof}
As a given source node is selected, the next nodes are added to the path with a given probability. 

Let us assume that ${\rm {\mathcal G}}_{m} $ is determined. Let us index the nodes by an $ID$ identifier tag, $ID=\left\{A,B,\ldots \right\}$, and let ${\rm S}_{I} $ be the set of unvisited neighbor nodes of a node $I$. 

Let $B_{F} \left(E_{{\rm L}_{l} } \left(I,J\right)\right)$ refer to the entanglement throughput of a given ${\rm L}_{l} $-level entangled link $E_{{\rm L}_{l} } \left(I,J\right)$ between nodes $\left(I,J\right)$ measured in the number of $d$-dimensional entangled states per second at a particular fidelity $F$ \cite{ref1, ref3, ref4}.  

Further, let $J$ be a neighbor node of $I$ with entangled connection $E_{{\rm L}_{l} } \left(I,J\right)$ and with entanglement throughput $B_{F} \left(E_{{\rm L}_{l} } \left(I,J\right)\right)$. 

A cost function, $\Omega \left(I,J\right)$, between nodes $\left(I,J\right)$ is defined as
\begin{equation} \label{ZEqnNum566366} 
\Omega \left(I,J\right)=\frac{1}{C\left(E_{{\rm L}_{l} } \left(I,J\right)\right)+\zeta _{J} } ,                                            
\end{equation} 
where $C\left(E_{{\rm L}_{l} } \left(I,J\right)\right)$ is the cost of entangled link $E_{{\rm L}_{l} } \left(I,J\right)$ defined as
\begin{equation} \label{ZEqnNum253059} 
C\left(E_{{\rm L}_{l} } \left(I,J\right)\right)=\frac{1}{B_{F} \left(E_{{\rm L}_{l} } \left(I,J\right)\right)} ,                                            
\end{equation} 
while $\zeta _{J} $ is the cost of quantum storage in node $J$. 

Let $\lambda _{E_{{\rm L}_{l} } \left(I,J\right)} $ be the entanglement utility coefficient of entangled link $E_{{\rm L}_{l} } \left(I,J\right)$ between nodes $I$ and $J$, initialized as $\lambda _{E_{{\rm L}_{l} } \left(I,J\right)} \ge 0$. This amount is equivalent to the utility of the entangled link $E_{{\rm L}_{l} } \left(I,J\right)$ that it has taken to arrive at the current node $J$ from $I$.

At a given $B_{F} \left(E_{{\rm L}_{l} } \left(I,J\right)\right)$, the initial $\lambda _{E_{{\rm L}_{l} } \left(I,J\right)} $ entanglement utility \cite{ref4} of link $E_{{\rm L}_{l} } \left(I,J\right)$ is updated to $\lambda '_{E_{{\rm L}_{l} } \left(I,J\right)} $ as
\begin{equation} \label{ZEqnNum640853} 
\begin{split}
  {{{{\lambda }'}}_{{{E}_{{{\text{L}}_{l}}}}\left( I,J \right)}}&={{\left( \frac{1}{{{\lambda }_{{{E}_{{{\text{L}}_{l}}}}\left( I,J \right)}}}+{{B}_{F}}\left( {{E}_{{{\text{L}}_{l}}}}\left( I,J \right) \right) \right)}^{-1}} \\ 
 & =\frac{{{\lambda }_{{{E}_{{{\text{L}}_{l}}}}\left( I,J \right)}}}{1+{{B}_{F}}\left( {{E}_{{{\text{L}}_{l}}}}\left( I,J \right) \right){{\lambda }_{{{E}_{{{\text{L}}_{l}}}}\left( I,J \right)}}}.  
\end{split}
\end{equation} 
Using these cost functions, the $\Pr \left(I,J\right)$ probability that from node $I$ a node $J$ is selected is as follows:
\begin{equation} \label{ZEqnNum614466} 
\Pr \left(I,J\right)=\left\{\begin{array}{l} {\frac{\left(\lambda '_{E_{{\rm L}_{l} } \left(I,J\right)} \right)^{{{\omega }^{*}}} \left(\Omega \left(I,J\right)\right)^{\delta } }{\sum _{\forall X\in {\rm S}_{I} }\left(\lambda '_{E_{{\rm L}_{l} } \left(I,X\right)} \right)^{{{\omega }^{*}}} \left(\Omega \left(I,X\right)\right)^{\delta }  } ,{\rm \; if\; }J\in {\rm S}_{I} } \\ {0,{\rm otherwise}} \end{array}\right. , 
\end{equation} 
where ${{\omega }^{*}}$ and $\delta $ are weighting coefficients \cite{ref4}. 

Using \eqref{ZEqnNum566366}, \eqref{ZEqnNum253059}, and \eqref{ZEqnNum614466} for all node pairs, the ${\rm {\mathcal M}}$ method to build a random ${\rm {\mathcal G}}_{et} $ entanglement throughput tree using a ${\rm {\mathcal G}}_{m} $ graph is as follows \cite{ref8, ref9}. 

Let ${\rm S}'$ refer to the set of already reached destination nodes, and let ${\rm {\mathcal I}}$ be the set of initial nodes, ${\rm {\mathcal F}}_{I} $ the set of feasible neighboring nodes to node $I$. Let ${\rm {\mathcal D}}$ be the set of destination nodes. 

The method is given in Procedure 1.  

\setcounter{algocf}{0}
\begin{proced}
  \DontPrintSemicolon
\caption{\textit{Random Entanglement Throughput Tree Construction}}
\textbf{Step 1}. Initialize ${\rm {\mathcal G}}_{et} =\emptyset $, ${\rm S}'=\emptyset $. Select a node $I$ of set ${\rm {\mathcal I}}$ and determine ${\rm {\mathcal F}}_{I} $. 

\textbf{Step 2}. If ${\rm {\mathcal F}}_{I} =\emptyset $, then remove node $I$ from ${\rm {\mathcal I}}$, as ${\rm {\mathcal I}}\, ={\rm {\mathcal I}}-I$, otherwise compute probability $\Pr \left(I,J\right)$ via \eqref{ZEqnNum614466} for each node $J$ of ${\rm {\mathcal F}}_{I} $. 

\textbf{Step 3}. Define uniformly distributed variable $x\in \left(0,1\right]$. If $x>0.5$, select that node $J$ from ${\rm {\mathcal F}}_{I} $ that has large $\Pr \left(I,J\right)$. If $x\le 0.5$, choose $J$ randomly. 

\textbf{Step 4}. Update sets ${\rm {\mathcal G}}_{et} $ and ${\rm {\mathcal I}}$ as ${\rm {\mathcal G}}_{et} ={\rm {\mathcal G}}_{et} \bigcup \left(I,J\right)$, ${\rm {\mathcal I}}={\rm {\mathcal I}}\; {\rm =}\bigcup \; J$. If node $J$ is a destination node, $J\in {\rm {\mathcal D}}$, then update set ${\rm S}'$ as ${\rm S}'{\rm }={\rm S}'\bigcup \; J$. 

\textbf{Step 5}. Update $\lambda '_{E_{{\rm L}_{l} } \left(I,J\right)} $ as $\lambda '_{E_{{\rm L}_{l} } \left(I,J\right)} =\left(1-\varphi \right)\lambda '_{E_{{\rm L}_{l} } \left(I,J\right)} +\varphi \lambda _{E_{{\rm L}_{l} } \left(I,J\right)} $, where $\varphi $ is an evolution parameter, $\varphi \in \left[0,1\right]$, and $\lambda _{E_{{\rm L}_{l} } \left(I,J\right)} $ is the initial value of entanglement utility. Merge nodes of ${\rm {\mathcal G}}_{et} $ that duplicate entangled connections, and check the reachability of the nodes of ${\rm {\mathcal D}}$.

\textbf{Step 6}. Repeat steps 2--5 until ${\rm {\mathcal I}}\ne \emptyset $ or ${\rm S}'\ne \emptyset $. Remove unused entangled links from ${\rm {\mathcal G}}_{et} $ and output ${\rm {\mathcal G}}_{et} $.
\end{proced}
The proof is concluded here.
\end{proof}

The structure of a ${\rm {\mathcal G}}_{et} $ entanglement throughput tree is illustrated in \fref{fig3}. 

 \begin{center}
\begin{figure*}[!h]
\begin{center}
\includegraphics[angle = 0,width=1\linewidth]{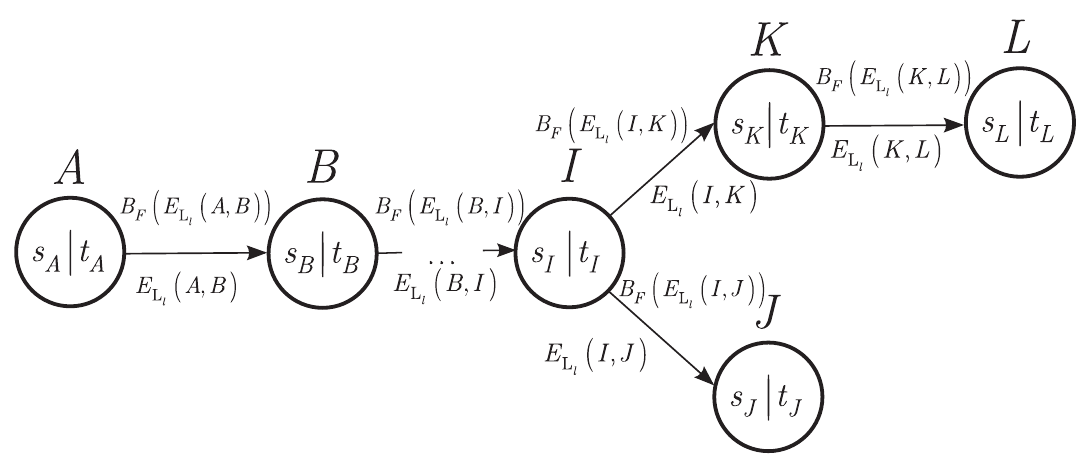}
\caption{The structure of a ${\rm {\mathcal G}}_{et} $ entanglement throughput tree. A quantum node has an $ID$ identifier tag, $ID=\left\{A,B,\ldots ,L\right\}$, and all incoming and outcoming entangled links are identified by the $s_{ID} $ source neighbor node and the $t_{ID} $ target neighbor node. Node $A$ represents a source, while the destination nodes are $J,K,L$. The $B_{F} \left(E_{{\rm L}_{l} } \left(x,y\right)\right)$ entanglement throughput of all $\left(x,y\right)$ node pairs are depicted above the directed lines; the link identifier $E_{{\rm L}_{l} } \left(x,y\right)$ is depicted under the links.} 
 \label{fig3}
 \end{center}
\end{figure*}
\end{center}
 
\subsection{Entanglement Assignment Cycle}

In this section, we propose a solution for an optimal assignment (scheduling) of stored entanglement called the entanglement assignment cycle, $\alpha _{{\rm {\mathcal G}}_{et} } $. The goal of $\alpha _{{\rm {\mathcal G}}_{et} } $ is to achieve a minimal overall storage time $t_{s}^{*} \left({\rm {\mathcal G}}_{et} \right)$ at a given ${\rm {\mathcal G}}_{et} $ entanglement throughput tree.

\begin{lemma}
An entanglement assignment cycle can be determined by a weighted graph coloring method.
\end{lemma}
\begin{proof} 
To determine the minimal overall storage time for a ${\rm {\mathcal G}}_{et} $ entanglement throughput tree, the ${\rm {\mathcal C}}_{{\rm {\mathcal G}}_{et} } $ conflict graph of that ${\rm {\mathcal G}}_{et} $ is constructed first. In the ${\rm {\mathcal C}}_{{\rm {\mathcal G}}_{et} } $ graph, each vertex corresponds to a directed link of ${\rm {\mathcal G}}_{et} $ (an entangled connection). An edge exists between two vertices of ${\rm {\mathcal C}}_{{\rm {\mathcal G}}_{et} } $, if only the vertices (entangled connections) have a conflict. A conflict occurs if two (stored) entangled connects are associated to the same physical link. The problem is therefore to associate each link of ${\rm {\mathcal G}}_{et} $ a $\alpha _{{\rm {\mathcal G}}_{et} } $ storage schedule (optimal assignment of stored entanglement), which includes the list of time slots when a given link can transmit the stored entangled states such that $\alpha _{{\rm {\mathcal G}}_{et} } $ total number of time units is minimized. Therefore, our goal is to determine what entangled connects of the given ${\rm {\mathcal G}}_{et} $ should be scheduled in which time unit, such that the total storage time is minimal in $\alpha _{{\rm {\mathcal G}}_{et} } $. 

Let $\tau _{n,t} \in \left\{0,1\right\}$ be an indicator variable, defined as
\begin{equation} \label{5)} 
{\tau }_{n,t}=\left\{ \begin{array}{l}
1,\text{ if }n\text{ is associated at } t \\ 
0,\text{ otherwise.} \end{array}
\right. 
\end{equation} 
For a periodic scheduling, 
\begin{equation} \label{6)} 
\tau _{n,t} =\tau _{n,t+iT} ,                                                     
\end{equation} 
where $T$ is a period and $i$ is a constant. 

For an entangled connection $n$, let us define $\wedge \left(n\right)$ as the set of entangled connects $n'$ that are scheduled in the same time unit $t$, but the physical link can transmit only $n$ or $n'$. As follows, for any $n'\in \wedge \left(n\right)$, 
\begin{equation} \label{7)} 
\tau _{n,t} +\tau _{n',t} \le 1.                                                   
\end{equation} 
As it can be concluded, this problem is analogous to the coloring of the conflict graph ${\rm {\mathcal C}}_{{\rm {\mathcal G}}_{et} } $. 

Then, let us assume that each entangled connection (entangled link) has a weight $w\left(n\right)$, which is defined as 
\begin{equation} \label{8)} 
w\left(n\right)=\left\{ \begin{array}{l}
1,\text{ if }F_i=F_{max} \\ 
\left\lceil \frac{F_{max}}{F_i}\right\rceil,\text{ if }F_i<F_{max} \end{array}
\right., 
\end{equation} 
where $F_{i} $ is the fidelity of entangled connection $i$ and $F_{\max } $ is the largest fidelity available. As follows, for a lower-fidelity entangled connection, the transmission of a given amount of information requires more time units.  

As the weights are determined, the problem is analogous to a weighted graph coloring of the conflict graph ${\rm {\mathcal C}}_{{\rm {\mathcal G}}_{et} } $, which is the assignment of at least $w\left(n\right)$ distinct colors to each entangled connection $n$ such that no two entangled connections sharing the same color interfere with each other on the physical link. For this purpose, a simple distributed weighted coloring algorithm \cite{ref10} can be straightforwardly applied. 

The method for the ${\rm {\mathcal W}}\left({\rm {\mathcal C}}_{{\rm {\mathcal G}}_{et} } \right)$ weighted coloring of conflict graph ${\rm {\mathcal C}}_{{\rm {\mathcal G}}_{et} } $ is summarized in Procedure 2. 

\begin{proced}
  \DontPrintSemicolon
\caption{\textit{Weighted Coloring of a Conflict Graph}}

\textbf{Step 1}. Determine conflict graph ${\rm {\mathcal C}}_{{\rm {\mathcal G}}_{et} } $ of ${\rm {\mathcal G}}_{et} $, and compute weights $w\left(n\right)$ for all $n$. Assign weight $w\left(n\right)$ to vertex $n\in {\rm {\mathcal C}}_{{\rm {\mathcal G}}_{et} } $.    

\textbf{Step 2}. Construct a new conflict graph ${\rm {\mathcal C}}'_{{\rm {\mathcal G}}_{et} } $ from ${\rm {\mathcal C}}_{{\rm {\mathcal G}}_{et} } $: for each vertex $n$ with weight $w\left(n\right)$, create $w\left(n\right)$ vertices, $\left\{n_{1} ,\ldots ,n_{w\left(n\right)} \right\}$, and add to ${\rm {\mathcal C}}'_{{\rm {\mathcal G}}_{et} } $. 

\textbf{Step 3}. Add to ${\rm {\mathcal C}}'_{{\rm {\mathcal G}}_{et} } $ the edges connecting $n_{a} $ and $n_{b} $, where $1\le a<b\le w\left(n\right)$.

\textbf{Step 4}. Add to ${\rm {\mathcal C}}'_{{\rm {\mathcal G}}_{et} } $ the edge between  $i_{a} $ and $j_{b} $ only if there is an edge between entangled connects $i$ and $j$ in ${\rm {\mathcal C}}_{{\rm {\mathcal G}}_{et} } $. Apply the unweighted vertex coloring algorithm \cite{ref10} on ${\rm {\mathcal C}}'_{{\rm {\mathcal G}}_{et} } $.

\textbf{Step 5}. Assign $n$ all the colors that are used by $n_{k} $, for $1\le k\le w\left(n\right)$ in graph ${\rm {\mathcal C}}'_{{\rm {\mathcal G}}_{et} } $. 

\textbf{Step 6}. Output the ${\rm {\mathcal W}}\left({\rm {\mathcal C}}_{{\rm {\mathcal G}}_{et} } \right)$ weighted coloring of conflict graph ${\rm {\mathcal C}}_{{\rm {\mathcal G}}_{et} } $. 

\end{proced}

After an entanglement assignment cycle has been determined by the weighted graph coloring method, in the next step the $\Delta \left({\rm {\mathcal W}}\left({\rm {\mathcal C}}_{{\rm {\mathcal G}}_{et} } \right)\right)$ time intervals between each time unit of a given cycle is computed. For this purpose, a linear programming method \cite{ref8} can be applied.
\end{proof} 

\section{Quantum-Layer Optimization}
\label{sec3}
In this section, we define the multi-objective optimization for the quantum layer. The multi-objective function covers a parallel optimization of quantum memory usage and entanglement throughput for a shortest path. 

\begin{theorem}
At a given ${\rm {\mathcal G}}_{m} $, the quantum layer of a quantum repeater network $N$ can be optimized by a procedure ${\rm {\mathcal P}}_{Q} $ that achieves a parallel minimization of the quantum memory usage time $t_{s} $, the maximization of entanglement throughput $B_{F} $, and the minimization of the number $\left|{\rm {\mathcal P}}\right|$ of entangled links between two arbitrary nodes.
\end{theorem}
\begin{proof}
The aim of the procedure is to find an optimal entanglement throughput tree ${\rm {\mathcal G}}_{et}^{*} $ for which the $t_{s} \left({\rm {\mathcal G}}_{et}^{*} \right)$ overall storage time is minimal, that is, 
\begin{equation} \label{ZEqnNum603596} 
t_{s} \left({\rm {\mathcal G}}_{et}^{*} \right)=t_{s}^{*} \left({\rm {\mathcal G}}_{et}^{*} \right);                                                  
\end{equation} 
the $B_{F} \left({\rm {\mathcal G}}_{et}^{*} \right)$ entanglement throughput for all links is maximal, 
\begin{equation} \label{ZEqnNum624672} 
B_{F} \left({\rm {\mathcal G}}_{et}^{*} \right)=B_{F}^{*} \left({\rm {\mathcal G}}_{et}^{*} \right);                                                
\end{equation} 
and the $\left|{\rm {\mathcal P}}\left({\rm {\mathcal G}}_{et}^{*} \right)\right|$ number of entangled links of a path is minimal, thus 
\begin{equation} \label{ZEqnNum206696} 
\left|{\rm {\mathcal P}}\left({\rm {\mathcal G}}_{et}^{*} \right)\right|=\left|{\rm {\mathcal P}}^{*} \left({\rm {\mathcal G}}_{et}^{*} \right)\right|,                                              
\end{equation} 
where ${\rm {\mathcal P}}^{*} $ is the shortest path.

The procedure is based on the fact that for a given ${\rm {\mathcal G}}_{et}^{*} $ with conflict graph ${\rm {\mathcal C}}_{{\rm {\mathcal G}}_{et}^{*} } $, from the knowledge of ${\rm {\mathcal W}}\left({\rm {\mathcal C}}_{{\rm {\mathcal G}}_{et}^{*} } \right)$ weighted coloring of conflict graph ${\rm {\mathcal C}}_{{\rm {\mathcal G}}_{et}^{*} } $, $\Delta \left({\rm {\mathcal W}}\left({\rm {\mathcal C}}_{{\rm {\mathcal G}}_{et}^{*} } \right)\right)$ time intervals between each time unit of a given cycle, the required objective values of \eqref{ZEqnNum603596}, \eqref{ZEqnNum624672}, and \eqref{ZEqnNum206696} can be determined.

Let ${\rm {\mathcal S}}_{{\rm {\mathcal G}}_{et} }^{*} $ be the set of optimal ${\rm {\mathcal G}}_{et}^{*} $ entanglement throughput trees. Then the aim of the procedure is to determine ${\rm {\mathcal S}}_{{\rm {\mathcal G}}_{et} }^{*} $. 

The problem can be rewritten via a solution set ${\rm X} $ with decision variables \cite{ref8} 
\begin{equation} \label{12)} 
{\rm X} =\left\{x_{1} ,\ldots ,x_{n} \right\},                                                   
\end{equation} 
where $n$ is the number of all links in a given quantum memory utilization graph ${\rm {\mathcal G}}_{m} $, and $x_{i} \in \left\{0,1\right\}$ is defined as
\begin{equation} \label{13)} 
x_{i} =\left\{\begin{array}{l} {1,{\rm \; if\; link\; }i{\rm \; of\; {\mathcal G}}_{m} {\rm \; is\; selected\; by\; {\mathcal M}}} \\ {0,{\rm \; otherwise.\; \; \; \; \; \; \; \; \; \; \; \; \; \; \; \; \; \; \; \; \; \; \; \; \; \; \; \; \; \; \; \; \; \; \; \; \; }} \end{array}\right.  
\end{equation} 
Let ${\rm X} _{{\rm {\mathcal G}}_{et} ,i} $ be a solution $i$ from solution set ${\rm X} _{{\rm {\mathcal G}}_{et} } $. Let ${\rm X} _{{\rm {\mathcal G}}_{et} ,j} \angle {\rm X} _{{\rm {\mathcal G}}_{et} ,i} $ refer to solution ${\rm X} _{{\rm {\mathcal G}}_{et} ,i} $ dominating solution ${\rm X} _{{\rm {\mathcal G}}_{et} ,j} $, which is 
\begin{equation} \label{14)} 
 \begin{split}
{\rm X}_{{\mathcal{G}}_{et},j}\angle {\rm X}_{{\mathcal{G}}_{et},i}:&\left\{t^*_s\left({\mathcal{G}}_{et},i\right)\le t^*_s\left({\mathcal{G}}_{et},j\right)\right., \\ 
&B^*_F\left({\mathcal{G}}_{et},i\right)\ge B^*_F\left({\mathcal{G}}_{et},j\right), \\ 
&\left.\left|{\mathcal{P}}^*\left({\mathcal{G}}_{et},i\right)\right|\le \left|{\mathcal{P}}^*\left({\mathcal{G}}_{et},j\right)\right|\right\}, \end{split}
\end{equation} 
and for these relations, there is at least strict inequality \cite{ref8}. If ${\rm X} _{{\rm {\mathcal G}}_{et} ,i} $ dominates all other possible solutions, then ${\rm X} _{{\rm {\mathcal G}}_{et} ,i} $ is a set of nondominated solutions. 

For each set ${\rm X} $, let $\kappa $ refer to the set that contains the best nondominated solutions that have been found at a particular iteration. If $\kappa $ changes, then the entanglement utilities of the links are updated. If $\kappa $ is stationary, then the elements of set $\kappa $ are used to update the entanglement utilities. The former serves as an improvement in the exploration process whereas the latter aims to yield more information via the best solutions \cite{ref8}.

At a given graph ${\rm {\mathcal G}}_{m} $ of quantum memory utilization, the procedure ${\rm {\mathcal P}}_{Q} $ for the optimization of the quantum layer is defined as follows. The output of ${\rm {\mathcal P}}_{Q} $ is an optimal set ${\rm {\mathcal S}}_{{\rm {\mathcal G}}_{et} }^{*} $ of entanglement throughput trees that realize the conditions of the multi-objective optimization function.
The steps of the quantum layer optimization are summarized in Algorithm 1.

\setcounter{algocf}{0}
\begin{algo}
  \DontPrintSemicolon
\caption{\textit{Quantum Layer Optimization}}

\textbf{Step 1}. Set ${\rm {\mathcal S}}_{{\rm {\mathcal G}}_{et} }^{*} =\emptyset $, and for all entangled connects, initialize $\lambda '_{E_{{\rm L}_{l} } \left(I,J\right)} $ as 
\[\lambda '_{E_{{\rm L}_{l} } \left(I,J\right)} =\lambda _{E_{{\rm L}_{l} } \left(I,J\right)} .\] 

\textbf{Step 2}. Determine ${\rm X} $, and apply method ${\rm {\mathcal M}}$ for building an ${\rm {\mathcal G}}_{et} $ entanglement throughput tree as ${\rm {\mathcal M}}\left({\rm X} \right)={\rm {\mathcal G}}_{et}^{{\rm X} } $. For a given ${\rm {\mathcal G}}_{et}^{{\rm X} } $, determine the $\alpha _{{\rm {\mathcal G}}_{et} } $ optimal assignment of stored entanglement, the $t_{s}^{*} $ overall minimal storage time, the $B_{F}^{*} $ maximal entanglement throughput, and the $\left|{\rm {\mathcal P}}^{*} \right|$ minimal number of entangled links, where ${\rm {\mathcal P}}^{*} $ is a shortest path. 

\textbf{Step 3}. If ${\rm X} $ is not dominated by any ${\rm X} _{{\rm {\mathcal G}}_{et} } \in \kappa $, that is, ${\rm X} _{{\rm {\mathcal G}}_{et} } \angle {\rm X} $ for any ${\rm X} _{{\rm {\mathcal G}}_{et} } $, then update $\kappa $ as 
\[\kappa =\kappa \bigcup {\rm X} -\left\{\left. {\rm X} _{{\rm {\mathcal G}}_{et} } \right|{\rm X} _{{\rm {\mathcal G}}_{et} } \angle {\rm X} \right\}.\] 

\textbf{Step 4}. If $\kappa $ has been updated to $\kappa '$, then update the entanglement utility as $\lambda '_{E_{{\rm L}_{l} } \left(I,J\right)} =\lambda _{E_{{\rm L}_{l} } \left(I,J\right)} $ for all entangled connects of ${\rm {\mathcal G}}_{m} $. 

\textbf{Step 5}. If $\kappa $ has not been updated, then for all ${\rm X} _{{\rm {\mathcal G}}_{et} } \in \kappa $ compute 
\[\Pi =\rho t_{s} \left({\rm {\mathcal G}}_{et}^{{\rm X} } \right)+\varsigma B_{F}^{*} \left({\rm {\mathcal G}}_{et}^{{\rm X} } \right)+\Upsilon \left|{\rm {\mathcal P}}^{*} \left({\rm {\mathcal G}}_{et}^{{\rm X} } \right)\right|,\] 
where $t_{s} \left({\rm {\mathcal G}}_{et}^{{\rm X} } \right)$ is the storage time associated to tree ${\rm {\mathcal G}}_{et}^{{\rm X} } $, $B_{F}^{*} \left({\rm {\mathcal G}}_{et}^{{\rm X} } \right)$ is the maximal entanglement throughput associated to tree ${\rm {\mathcal G}}_{et}^{{\rm X} } $, $\left|{\rm {\mathcal P}}^{*} \left({\rm {\mathcal G}}_{et}^{{\rm X} } \right)\right|$ is the number of entangled links of the ${\rm {\mathcal P}}^{*} \left({\rm {\mathcal G}}_{et}^{{\rm X} } \right)$ shortest path of ${\rm {\mathcal G}}_{et}^{{\rm X} } $, while $\rho $, $\varsigma $, and $\Upsilon $ are weighting coefficients. For all entangled links of graph ${\rm {\mathcal G}}_{m} $, update $\lambda '_{E_{{\rm L}_{l} } \left(I,J\right)} $ as 
\[\lambda '_{E_{{\rm L}_{l} } \left(I,J\right)} =\left(1-\Pr \left(I,J\right)\right)\lambda '_{E_{{\rm L}_{l} } \left(I,J\right)} +\Pr \left(I,J\right)\Pi .\] 

\textbf{Step 6}. Output optimal entanglement throughput tree set ${\rm {\mathcal S}}_{{\rm {\mathcal G}}_{et} }^{*} $.

\end{algo}

Note if there are no nondominated solutions, then the values of the weighting coefficients $\rho $, $\varsigma $, and $\Upsilon $ in Step 5 of Algorithm 1 can be selected according to the actual trade-off requirements. 

The proof is therefore concluded here.
\end{proof}
 
\section{Classical-Layer Optimization}
\label{sec4}
 In this section we characterize the classical layer optimization procedure.

\begin{theorem}
The cost function of classical communications can be minimized by a procedure ${\rm {\mathcal P}}_{C} $. 
\end{theorem}
\begin{proof}
Let 
\begin{equation}
\begin{split}
  & {{f}_{t_{s}^{*},B_{F}^{*},\left| {{\mathcal{P}}^{*}} \right|}}\left( {{\Theta }_{i}} \right) \\ 
 & =\left( \sum\limits_{g=1}^{{{N}_{t_{s}^{*}}}}{{{S}_{t_{s}^{*},g}}\left( i \right)+c_{{{N}_{t_{s}^{*},g}}}^{L}\left( i \right)} \right)+\left( \sum\limits_{g=1}^{{{N}_{B_{F}^{*}}}}{{{S}_{B_{F}^{*},g}}\left( i \right)+c_{{{N}_{B_{F}^{*},g}}}^{L}\left( i \right)} \right)+\left( \sum\limits_{g=1}^{{{N}_{\left| {{\mathcal{P}}^{*}} \right|}}}{{{S}_{\left| {{\mathcal{P}}^{*}} \right|,g}}\left( i \right)+c_{{{N}_{\left| {{\mathcal{P}}^{*}} \right|}},g}^{L}\left( i \right)} \right) \\ 
\end{split}
\end{equation}
be the cost function of classical communication of the multilayer optimization procedure, where $\Theta _{i} \in {\rm {\rm R}}^{p} $ is a $p$-dimensional real vector of an $i$-th system state of the quantum network, $N_{t_{s}^{*} } $, $N_{B_{F}^{*} } $, $N_{\left|{\rm {\mathcal P}}^{*} \right|} $ are the number of nodes that require the determination of optimal $t_{s}^{*} $, $B_{F}^{*} $ and $\left|{\rm {\mathcal P}}^{*} \right|$, $S_{t_{s}^{*} ,g} \left(i\right)$, $S_{B_{F}^{*} ,g} \left(i\right)$, $S_{\left|{\rm {\mathcal P}}^{*} \right|,g} \left(i\right)$ refer to the number of classical steps required to find $t_{s}^{*} $, $B_{F}^{*} $ and $\left|{\rm {\mathcal P}}^{*} \right|$ for a particular node $g$ of network $N$, while $c_{N_{t_{s}^{*} ,g} }^{L} \left(i\right)$, $c_{N_{B_{F}^{*} ,g} }^{L} \left(i\right)$, and $c_{N_{\left|{\rm {\mathcal P}}^{*} \right|,g} }^{L} \left(i\right)$ are the costs of classical link $L$ used for the determination of $t_{s}^{*} $, $B_{F}^{*} $ and $\left|{\rm {\mathcal P}}^{*} \right|$ for a given $g$.

Then, let introduce indices for $\Theta _{i} $ as
\begin{equation} \label{16)} 
\Theta _{i} \left(j,k,l\right),                                                    
\end{equation} 
where $j$ is the index of a desired optimal system state, $k$ is the index of an optimal system state reproduction step, $l$ is the index of a non-optimal system state event \cite{ref11}.  

Let assume that there is a set of $S$ sub-states $\left\{\Theta _{1} ,\ldots ,\Theta _{S} \right\}$ in the network, then the $T$ total network state is evaluated as
\begin{equation} \label{17)} 
T\left(j,k,l\right)=\left\{\Theta _{i} \left. \left(j,k,l\right)\right|i=1,\ldots ,S\right\}.                                      
\end{equation} 
Let $f_{t_{s}^{*} ,B_{F}^{*} ,\left|{\rm {\mathcal P}}^{*} \right|} \left(\Theta _{i} \left. \left(j,k,l\right)\right|\right)$ be the cost function of classical communication at a given $\Theta _{i} \left. \left(j,k,l\right)\right|$. Then, if 
\begin{equation} \label{18)} 
f_{t_{s}^{*} ,B_{F}^{*} ,\left|{\rm {\mathcal P}}^{*} \right|} \left(\Theta _{i} \left(j+1,k,l\right)\right)<f_{t_{s}^{*} ,B_{F}^{*} ,\left|{\rm {\mathcal P}}^{*} \right|} \left(\Theta _{i} \left(j,k,l\right)\right),                          
\end{equation} 
then the system state evolves from $j$ to $j+1$ as 
\begin{equation} \label{ZEqnNum287477} 
\Theta _{i} \left(j+1,k,l\right)=\Theta _{i} \left(j,k,l\right)+c\left(i\right)u\left(j\right),                             
\end{equation} 
where $c\left(i\right)$ is the number of random system states, while $u\left(j\right)$ quantifies a unit cost of system change \cite{ref11, ref12}.   

Therefore, for a set of $S$ sub-states $\left\{\Theta _{1} ,\ldots ,\Theta _{S} \right\}$ and current vector $\Theta \in {\mathbb{R} }^{p} $, the $C_{N} $ total cost of classical communication is yielded as
\begin{equation} \label{20)} 
C_{N} \left(\Theta ,\Theta _{i} \left(j,k,l\right)\right)=\sum _{i=1}^{S}f_{t_{s}^{*} ,B_{F}^{*} ,\left|{\rm {\mathcal P}}^{*} \right|} \left(\Theta _{i} \left(j,k,l\right)\right),i=1,\ldots ,S ,                    
\end{equation} 
which can be rewritten \cite{ref11} as 
\begin{equation} \label{ZEqnNum535604} 
\begin{split}
  & {{C}_{N}}\left( \Theta ,T\left( j,k,l \right) \right) \\ 
 & =\sum\limits_{i=1}^{S}{\left( -\text{A} {{e}^{-{{R}_{\text{A}}}}}\sum\limits_{m=1}^{p}{{{\left( {{\Theta }^{\left( m \right)}}-\Theta _{i}^{\left( m \right)} \right)}^{2}}} \right)}+\sum\limits_{i=1}^{S}{\left( \nu {{e}^{-{{R}_{\nu }}}}\sum\limits_{m=1}^{p}{{{\left( {{\Theta }^{\left( m \right)}}-\Theta _{i}^{\left( m \right)} \right)}^{2}}} \right)}  
\end{split}
\end{equation} 
where ${\rm A} $ is the distribution-entity of a current system state, $R_{{\rm A} } $ is the information transmission rate of ${\rm A} $, $\nu $ is the distribution-entity of a system state, $R_{\nu } $ is the information transmission rate of $\nu $, while $\Theta ^{\left(m\right)} $ is the $m$-th element of a current network state vector $\Theta $, and $\Theta _{i}^{\left(m\right)} $ is the $m$-th element of $\Theta _{i} $.

Using $C_{N} \left(\Theta ,T\left(j,k,l\right)\right)$ (see \eqref{ZEqnNum535604}), an environment-dependent cost function $C_{e} $ is defined as
\begin{equation} \label{ZEqnNum133844} 
C_{e} =C_{N} \left(\Theta ,T\left(j,k,l\right)\right)e^{\left(M-f_{t_{s}^{*} ,B_{F}^{*} ,\left|{\rm {\mathcal P}}^{*} \right|} \left(\Theta _{i} \left(j,k,l\right)\right)\right)} , 
\end{equation} 
where $M$ is a tuning parameter \cite{ref11}. 

From \eqref{ZEqnNum133844}, the $F_{cost}^{i} $ cost function at a given $\Theta _{i} \left(j,k,l\right)$ is defined as
\begin{equation} \label{23)} 
F_{cost}^{i} =f_{t_{s}^{*} ,B_{F}^{*} ,\left|{\rm {\mathcal P}}^{*} \right|} \left(\Theta _{i} \left(j,k,l\right)\right)+C_{e} \left(\Theta _{i} \left(j,k,l\right)\right).                              
\end{equation} 
Using the proposed basic model, we define an optimization procedure ${\rm {\mathcal P}}_{C} $ of the classical-layer to achieve a ${{C}_{N}}\left( \Theta ,T\left( j,k,l \right) \right) $ minimized cost function. 
The method is summarized in Algorithm 2.

\begin{algo}
  \DontPrintSemicolon
\caption{\textit{Classical Layer Optimization}}
\textbf{Step 1}. For an $i$-th system state $\Theta _{i} \left(j,k,l\right)$ initialize cost function
\[f_{t_{s}^{*} ,B_{F}^{*} ,\left|{\rm {\mathcal P}}^{*} \right|} \left(\Theta _{i} \left(j,k,l\right)\right)=f_{t_{s}^{*} ,B_{F}^{*} ,\left|{\rm {\mathcal P}}^{*} \right|} \left(\Theta _{i} \left(j,k,l\right)\right)+C_{N} \left(\Theta _{i} \left(j,k,l\right),T\left(j,k,l\right)\right).\] 
Define a random network state vector $\partial \left(i\right)\in {\mathbb{R} }^{p} $, where the $\partial _{m} \left(i\right)$ $m$-th element of $\partial \left(i\right)$ is a uniformly distributed number from the range of $\left[-1,1\right]$. 

\textbf{Step 2}. From $\Theta _{i} \left(j,k,l\right)$, define $\Theta _{i} \left(j+1,k,l\right)$ as
\[\Theta _{i} \left(j+1,k,l\right)=\Theta _{i} \left(j,k,l\right)+c\left(i\right){\textstyle\frac{\partial \left(i\right)}{\sqrt{\left(\partial \left(i\right)\right)^{T} \partial \left(i\right)} }} ,\] 
where $c\left(i\right)$ is the number of random system states. 

\textbf{Step 3}. Using $f_{t_{s}^{*} ,B_{F}^{*} ,\left|{\rm {\mathcal P}}^{*} \right|} \left(\Theta _{i} \left(j,k,l\right)\right)$, determine $f_{t_{s}^{*} ,B_{F}^{*} ,\left|{\rm {\mathcal P}}^{*} \right|} \left(\Theta _{i} \left(j+1,k,l\right)\right)$ as

$f_{t_{s}^{*} ,B_{F}^{*} ,\left|{\rm {\mathcal P}}^{*} \right|} \left(\Theta _{i} \left(j+1,k,l\right)\right)=f_{t_{s}^{*} ,B_{F}^{*} ,\left|{\rm {\mathcal P}}^{*} \right|} \left(\Theta _{i} \left(j,k,l\right)\right)+C_{N} \left(\Theta _{i} \left(j+1,k,l\right),T\left(j+1,k,l\right)\right).$ If 
\[f_{t_{s}^{*} ,B_{F}^{*} ,\left|{\rm {\mathcal P}}^{*} \right|} \left(\Theta _{i} \left(j+1,k,l\right)\right)<f_{t_{s}^{*} ,B_{F}^{*} ,\left|{\rm {\mathcal P}}^{*} \right|} \left(\Theta _{i} \left(j,k,l\right)\right),\] 
update $f_{t_{s}^{*} ,B_{F}^{*} ,\left|{\rm {\mathcal P}}^{*} \right|} \left(\Theta _{i} \left(j,k,l\right)\right)$ as 
\[f_{t_{s}^{*} ,B_{F}^{*} ,\left|{\rm {\mathcal P}}^{*} \right|} \left(\Theta _{i} \left(j,k,l\right)\right)=f_{t_{s}^{*} ,B_{F}^{*} ,\left|{\rm {\mathcal P}}^{*} \right|} \left(\Theta _{i} \left(j+1,k,l\right)\right).\] 
Update $\Theta _{i} \left(j+1,k,l\right)$ as
\[\Theta _{i} \left(j+1,k,l\right)=\Theta _{i} \left(j,k,l\right)+c\left(i\right){\textstyle\frac{\partial \left(i\right)}{\sqrt{\left(\partial \left(i\right)\right)^{T} \partial \left(i\right)} }} ,\] 
and compute $f_{t_{s}^{*} ,B_{F}^{*} ,\left|{\rm {\mathcal P}}^{*} \right|} \left(\Theta _{i} \left(j+1,k,l\right)\right)$. 

\textbf{Step 4}. Increase $i$, $i=i+1$. If $i<S$, apply steps 1-3 for the current network sub-state. If $j<n_{G\left(\tilde{\Theta }\right)} $, where $n_{G\left(\tilde{\Theta }\right)} $ is the total number of iteration steps on $j$ to reach an optimal global state $G(\tilde{\Theta })$, then increase $j$, $j=j+1$. 

\textbf{Step 5}. For all $i$,  determine 
\[F_{cost}^{i} =\sum _{j=1}^{n_{G\left(\tilde{\Theta }\right)} +1}f_{t_{s}^{*} ,B_{F}^{*} ,\left|{\rm {\mathcal P}}^{*} \right|} \left(\Theta _{i} \left(j,k,l\right)\right) +C_{e} \left(\Theta _{i} \left(j,k,l\right)\right).\] 
Remove system states for which $F_{cost}^{i} \ge \chi _{C} $, where $\chi _{C} $ is a threshold on $F_{cost}^{i} $ to get the minimized cost $F_{cost}^{\min } $ as
\[F_{cost}^{\min } =\sum _{i=1}^{S-\wp }F_{cost}^{i}  ,\] 
where $\wp $ is the total number of removed system states. For a given $E_{k} $ and $E_{l} $ expected values of $k$ and $l$, if $k<E_{k} $ then increase $k$, $k=k+1$, and if $l<E_{l} $ increase $l$, $l=l+1$.  
\end{algo}
These results conclude the proof.
 \end{proof}
 
\subsection{Large-Constrained Optimization}

The optimization efficiency can be further improved for a large constrained network scenario. This step can be replayed by a different approach, which allows an  optimization of the classical layer for an arbitrary constrained setting. The solution is based on the idea of system state merging. 
\begin{lemma} 
The optimization of the classical-layer can be extended to a large-constrained optimization by state merging. 
\end{lemma}
\begin{proof} 
An extended model of the classical-layer optimization is as follows. 

Let $N_{t_{s}^{*} } $, $N_{B_{F}^{*} } $, $N_{\left|{\rm {\mathcal P}}^{*} \right|} $ be the number of nodes that require the determination of optimal $t_{s}^{*} $, $B_{F}^{*} $ and $\left|{\rm {\mathcal P}}^{*} \right|$. Then, let $J$ be the objective function \cite{ref13}, as 
\begin{equation} \label{ZEqnNum495559} 
J=\min \sum\limits_{t=1}^{T}{\left( {{\alpha }_{t_{s}^{*}}}\left( t \right)+{{\alpha }_{B_{F}^{*}}}\left( t \right)+{{\alpha }_{\left| {{\mathcal{P}}^{*}} \right|}}\left( t \right) \right)},
\end{equation} 
where 
\begin{equation}
{{\alpha }_{t_{s}^{*}}}\left( t \right)=\sum\limits_{g=1}^{{{N}_{t_{s}^{*}}}}{{{S}_{t_{s}^{*},g}}\left( t \right)}+c_{{{N}_{t_{s}^{*},g}}}^{L}\left( t \right),
\end{equation}
\begin{equation}
{{\alpha }_{B_{F}^{*}}}\left( t \right)=\sum\limits_{g=1}^{{{N}_{B_{F}^{*}}}}{{{S}_{B_{F}^{*},g}}\left( t \right)}+c_{{{N}_{B_{F}^{*},g}}}^{L}\left( t \right),
\end{equation}
and
\begin{equation}
{{\alpha }_{\left| {{\mathcal{P}}^{*}} \right|}}\left( t \right)=\sum\limits_{g=1}^{{{N}_{\left| {{\mathcal{P}}^{*}} \right|}}}{{{S}_{\left| {{\mathcal{P}}^{*}} \right|,g}}\left( t \right)}+c_{{{N}_{\left| {{\mathcal{P}}^{*}} \right|,g}}}^{L}\left( t \right),
\end{equation}
where the coefficients $S_{t_{s}^{*} ,g} \left(t\right)$, $S_{B_{F}^{*} ,g} \left(t\right)$, $S_{\left|{\rm {\mathcal P}}^{*} \right|,g} \left(t\right)$ refer to the number of classical steps required to find $t_{s}^{*} $, $B_{F}^{*} $ and $\left|{\rm {\mathcal P}}^{*} \right|$ at a particular network node $g$ and time $t$, $t=1,\ldots ,T$, while $c_{N_{t_{s}^{*} ,g} }^{L} \left(t\right)$, $c_{N_{B_{F}^{*} ,g} }^{L} \left(t\right)$, $c_{N_{\left|{\rm {\mathcal P}}^{*} \right|,g} }^{L} \left(t\right)$ are the costs of classical link $L$ at a particular $g$ and $t$.
Let assume that $\Theta _{a} \left(j,k,l\right)$, $\Theta _{b} \left(j,k,l\right)$ and $\Theta _{c} \left(j,k,l\right)$ are some system states of the classical layer subject of state merging. From these sub-states, the $\Theta _{M} \left(j,k,l\right)$ merged system is defined as
\begin{equation} \label{ZEqnNum904315} 
\Theta _{M} \left(j,k,l\right)=\Theta _{a} \left(j,k,l\right)+\Phi \left(\Theta _{b} \left(j,k,l\right)-\Theta _{c} \left(j,k,l\right)\right),                    
\end{equation} 
where $\Phi $ is a merging factor, $\Phi \in \left[0,1\right]$. 

Using \eqref{ZEqnNum904315}, an $i$-th system state $\Theta _{i} \left(j,k,l\right)$ is updated as
\begin{equation} \label{26)} 
\Theta _{i} \left(j,k,l\right)=\left(\tilde{\Theta }_{i}^{a} \left(j,k,l\right),\tilde{\Theta }_{i}^{b} \left(j,k,l\right),\tilde{\Theta }_{i}^{c} \left(j,k,l\right)\right)^{T} ,                           
\end{equation} 
where
\begin{equation} \label{27)} 
\tilde{\Theta }_{i}^{a} \left(j,k,l\right)=\left\{\begin{array}{l} {\Theta _{i} \left(j,k,l\right),{\rm \; if\; }x_{a} >u} \\ {\Theta _{M} \left(j,k,l\right),{\rm otherwise}} \end{array}\right. , 
\end{equation} 
\begin{equation} \label{28)} 
\tilde{\Theta }_{i}^{b} \left(j,k,l\right)=\left\{\begin{array}{l} {\Theta _{i} \left(j,k,l\right),{\rm \; if\; }x_{b} >u} \\ {\Theta ^{*} \left(j,k,l\right),{\rm otherwise}} \end{array}\right. , 
\end{equation} 
and
\begin{equation} \label{29)} 
\tilde{\Theta }_{i}^{c} \left(j,k,l\right)=\left\{\begin{array}{l} {\Theta ^{*} \left(j,k,l\right),{\rm \; if\; }x_{c} >u} \\ {\Theta _{M} \left(j,k,l\right),{\rm otherwise}} \end{array}\right. , 
\end{equation} 
where $u$ is a uniform random number, $x_{a} ,x_{b} $ and $x_{c} $ are random numbers, $x_{a} ,x_{b} ,x_{c} \in \left[0,1\right]$, $\Theta ^{*} \left(j,k,l\right)$ is a system state that minimizes the objective function $J$. Then for $Q\in \left(a,b,c\right)$, the objective function value $J(\tilde{\Theta }_{i}^{Q} (j,k,l))$ is compared with the objective function $J\left(\Theta _{i} \left(j,k,l\right)\right)$ of $\Theta _{i} \left(j,k,l\right)$, and $\tilde{\Theta }_{i}^{Q} \left(j,k,l\right)$ is updated by the following rule: 
\begin{equation} \label{30)} 
{\widetilde{\Theta }}^Q_i\left(j,k,l\right)=\left\{ \begin{array}{l}
{\Theta }_i\left(j,k,l\right),\text{ if }J\left({\widetilde{\Theta }}^Q_i\left(j,k,l\right)\right)<J\left({\Theta }_i\left(j,k,l\right)\right) \\ 
{\widetilde{\Theta }}^Q_i\left(j,k,l\right),\text{ otherwise.} \end{array}
\right. 
\end{equation} 
The proof is concluded here.
\end{proof}

\subsection{Cost Uncertainty of Large-Scaled Optimization}
The determination of the minimal cost function \eqref{ZEqnNum495559} is can be approached by an output variable $O_{J} =f\left(\psi _{in} \right)$, where $\psi _{in} $ is the set of input variables, and $f\left(\cdot \right)$ is a function that transfers the uncertainty from the independent input random variables $\psi _{in} $ to the output variable $O_{J} $ \cite{ref13}. The output set $O_{J} $ can be rewritten as
\begin{equation} \label{1)} 
O_{J} =f\left(q,w_{1} ,\ldots ,w_{z} \right),                                               
\end{equation} 
where $q$ is the set of certain variables, while $w_{i} $ is an input variable under certainty with probability function $\delta _{f_{w_{i} } } $. 

Then, for a given variable $w_{i} $, two $pc\left(w_{i} \right)$ probability concentrations \cite{ref13}, $pc^{\left(1\right)} \left(w_{i} \right)$ and $pc^{\left(2\right)} \left(w_{i} \right)$ are defined as
\begin{equation}
p{{c}^{\left( 1 \right)}}\left( {{w}_{i}} \right)=\left( {{w}_{i,1}},{{\zeta }_{i,1}} \right)
\end{equation}
and
\begin{equation}
p{{c}^{\left( 2 \right)}}\left( {{w}_{i}} \right)=\left( {{w}_{i,2}},{{\zeta }_{i,2}} \right),
\end{equation}
where $w_{i,g} $, is the poth location of $w_{i} $ \cite{ref13}, $g=1,2$, while $\zeta _{i,g} $ is a weighting factor. 

Therefore, at a given $\left(i,g\right)$ parameterization, where $i=1,\ldots ,z$ and $g=1,2$, $O_{J} $ is expressed by variable $O_{J}^{\left(i,g\right)} $ as
\begin{equation} \label{3)} 
O_{J}^{\left(i,g\right)} =f\left(q,w_{i,g} ,\mu _{w_{1} } ,\mu _{w_{2} } ,\ldots ,\mu _{w_{z} } \right), 
\end{equation} 
where $\mu _{w_{i} } $ is the mean of $w_{i} $, while $w_{i,1} $ and $w_{i,2} $ are the poth locations of $w_{i} $. Therefore, the problem is reduced to $2z$ deterministic equations \cite{ref13}, from which the mean and standard deviation of the output random variable $O_{J}^{\left(i,g\right)} $ can be computed.

\section{Performance Evaluation}
\label{sec5}
In this section we study the performance of the proposed quantum layer and classical layer optimization methods. 

\subsection{Quantum Layer Optimization}

To study the convergence of the quantum layer optimization, we characterize the $D\left(\cdot \right)$ closeness function of the elements of the solution set from an optimal Pareto front, and the $\zeta \left(\cdot \right)$ ratio of optimal solutions in the solution set.

Let ${\rm X} _{{\rm {\mathcal G}}_{et} ,i} =\left\{t_{s}^{*} \left({\rm {\mathcal G}}_{et} ,i\right),B_{F}^{*} \left({\rm {\mathcal G}}_{et} ,i\right),\left|{\rm {\mathcal P}}^{*} \left({\rm {\mathcal G}}_{et} ,i\right)\right|\right\}$ be an $i$-th solution  from the solution set ${\rm X} _{{\rm {\mathcal G}}_{et} } \left(N_{it} \right)$ found at a $N_{it} $ finite number of iterations, and let ${\rm X} _{{\rm {\mathcal G}}_{et} }^{\infty } =\left\{t_{s}^{*} \left({\rm {\mathcal G}}_{et} ,z\right),B_{F}^{*} \left({\rm {\mathcal G}}_{et} ,z\right),\left|{\rm {\mathcal P}}^{*} \left({\rm {\mathcal G}}_{et} ,z\right)\right|\right\}$ refer to a solution $z$ from ${\rm X} _{{\rm {\mathcal G}}_{et} } \left(N_{it} \to \infty \right)$ at $N_{it} \to \infty $. Then, let $D\left({\rm X} _{{\rm {\mathcal G}}_{et} ,i} ,{\rm X} _{{\rm {\mathcal G}}_{et} }^{\infty } \right)$ be the distance on the Pareto front \cite{ref8, ref9, ref10} between ${\rm X} _{{\rm {\mathcal G}}_{et} ,i} $ and ${\rm X} _{{\rm {\mathcal G}}_{et} }^{\infty } $, defined as
\begin{equation} \label{ZEqnNum594299} 
D\left({\rm X} _{{\rm {\mathcal G}}_{et} ,i} ,{\rm X} _{{\rm {\mathcal G}}_{et} }^{\infty } \right)={\tfrac{1}{\chi }} \left(A+B+C\right), 
\end{equation} 
with $D\left({\rm X} _{{\rm {\mathcal G}}_{et} ,i} ,{\rm X} _{{\rm {\mathcal G}}_{et} }^{\infty } \right)\in \left[0,1\right]$, and where
\begin{equation} \label{2)} 
A=\frac{\left|t_{s}^{*} \left({\rm {\mathcal G}}_{et} ,z\right)-t_{s}^{*} \left({\rm {\mathcal G}}_{et} ,i\right)\right|}{t_{s}^{*} \left({\rm {\mathcal G}}_{et} ,z\right)} ,                                                             
\end{equation} 
\begin{equation} \label{3)} 
B=\frac{\left|B_{F}^{*} \left({\rm {\mathcal G}}_{et} ,z\right)-B_{F}^{*} \left({\rm {\mathcal G}}_{et} ,i\right)\right|}{B_{F}^{*} \left({\rm {\mathcal G}}_{et} ,z\right)} ,                                                         
\end{equation} 
and
\begin{equation} \label{4)} 
C=\frac{\left|\left|{\rm {\mathcal P}}^{*} \left({\rm {\mathcal G}}_{et} ,z\right)\right|-\left|{\rm {\mathcal P}}^{*} \left({\rm {\mathcal G}}_{et} ,i\right)\right|\right|}{\left|{\rm {\mathcal P}}^{*} \left({\rm {\mathcal G}}_{et} ,z\right)\right|} ,                                                         
\end{equation} 
while $\chi $ is a control parameter, $\chi >0$.

Then, let $\zeta \left({\rm X} _{{\rm {\mathcal G}}_{et} } \left(N_{it} \right),{\rm X} _{{\rm {\mathcal G}}_{et} } \left(N_{it} \to \infty \right)\right)$ be a ratio of the solution sets ${\rm X} _{{\rm {\mathcal G}}_{et} } \left(N_{it} \right)$ and ${\rm X} _{{\rm {\mathcal G}}_{et} } \left(N_{it} \to \infty \right)$ found at finite $N_{it} $ and $N_{it} \to \infty $ as
\begin{equation} \label{ZEqnNum872057} 
\zeta \left({\rm X} _{{\rm {\mathcal G}}_{et} } \left(N_{it} \right),{\rm X} _{{\rm {\mathcal G}}_{et} } \left(N_{it} \to \infty \right)\right)={\tfrac{card\left({\rm X} _{{\rm {\mathcal G}}_{et} } \left(N_{it} \right)\wedge {\rm X} _{{\rm {\mathcal G}}_{et} } \left(N_{it} \to \infty \right)\right)}{card\left({\rm X} _{{\rm {\mathcal G}}_{et} } \left(N_{it} \right)\right)}} ,                             
\end{equation} 
where $card\left({\rm X} _{{\rm {\mathcal G}}_{et} } \left(N_{it} \right)\right)$ is the cardinality of set ${\rm X} _{{\rm {\mathcal G}}_{et} } \left(N_{it} \right)$,  while ${\rm X} _{{\rm {\mathcal G}}_{et} } \left(N_{it} \right)\wedge {\rm X} _{{\rm {\mathcal G}}_{et} } \left(N_{it} \to \infty \right)$ is the intersection of solution sets ${\rm X} _{{\rm {\mathcal G}}_{et} } \left(N_{it} \right)$ and ${\rm X} _{{\rm {\mathcal G}}_{et} } \left(N_{it} \to \infty \right)$ (solutions included by both solution sets). 
Note, since the $\chi $ control parameter in \eqref{ZEqnNum594299} determines the $D$ distance from the optimal set, $\chi $ also affects the ratio of the solution sets \eqref{ZEqnNum872057}. If the value of $\chi $ is high, the $D$ distance in \eqref{ZEqnNum594299} is low, that results in a high cardinality of the intersection set in \eqref{ZEqnNum872057}. If $\chi $ is low, the $D$ distance in \eqref{ZEqnNum594299} is high, therefore the cardinality of the intersection set is small. 

The quantity of \eqref{ZEqnNum594299} for various $N_{it} $ and $\chi $ are depicted in \fref{fig_a1}. From the results it can be concluded that as $N_{it} $ increases, the $D\left({\rm X} _{{\rm {\mathcal G}}_{et} ,i} ,{\rm X} _{{\rm {\mathcal G}}_{et} }^{\infty } \right)$ distance significantly decreases, while the speed of convergence is controllable by $\chi $. It also can be verified that the $\chi $ control parameter has a significant impact on $D\left({\rm X} _{{\rm {\mathcal G}}_{et} ,i} ,{\rm X} _{{\rm {\mathcal G}}_{et} }^{\infty } \right)$, and the $D\left({\rm X} _{{\rm {\mathcal G}}_{et} ,i} ,{\rm X} _{{\rm {\mathcal G}}_{et} }^{\infty } \right)$ distance can be made arbitrarily small via a moderate value for $N_{it} $. 

 \begin{center}
\begin{figure*}[!h]
\begin{center}
\includegraphics[angle = 0,width=0.9\linewidth]{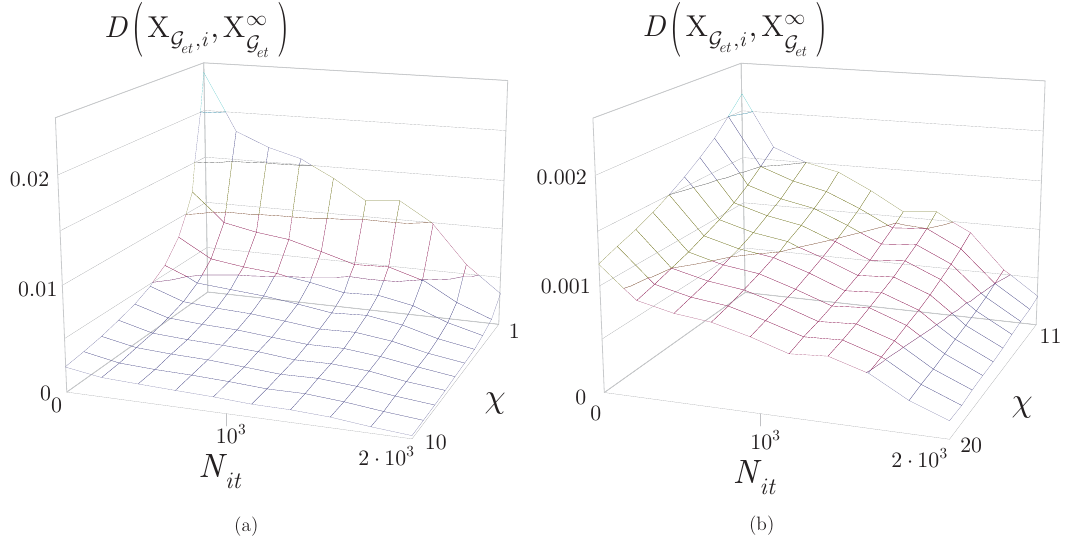}
\caption{The $D\left({\rm X} _{{\rm {\mathcal G}}_{et} ,i} ,{\rm X} _{{\rm {\mathcal G}}_{et} }^{\infty } \right)$ distance in function of $N_{it} $ and $\chi $. (a): The values of $D\left({\rm X} _{{\rm {\mathcal G}}_{et} ,i} ,{\rm X} _{{\rm {\mathcal G}}_{et} }^{\infty } \right)$ for $0<N_{it} \le 2\cdot 10^{3} $ and $1\le \chi \le 10$. (b): The values of $D\left({\rm X} _{{\rm {\mathcal G}}_{et} ,i} ,{\rm X} _{{\rm {\mathcal G}}_{et} }^{\infty } \right)$ for $0<N_{it} \le 2\cdot 10^{3} $ and $11\le \chi \le 20$.} 
 \label{fig_a1}
 \end{center}
\end{figure*}
\end{center}

In \fref{fig_a2}, the quantity of \eqref{ZEqnNum872057} is illustrated for different values of $N_{it} $. As $N_{it} $ increases the $\zeta \left({\rm X} _{{\rm {\mathcal G}}_{et} } \left(N_{it} \right),{\rm X} _{{\rm {\mathcal G}}_{et} } \left(N_{it} \to \infty \right)\right)$ ratio increases significantly, while the $\chi $ control parameter has a moderate effect on the ratio of \eqref{ZEqnNum872057}. The ratio exceeds $0.5$ at $N_{it} \approx 0.5\cdot 10^{3} $, and can be increased to arbitrarily high via a moderate increment in $N_{it} $.

 \begin{center}
\begin{figure*}[!h]
\begin{center}
\includegraphics[angle = 0,width=0.9\linewidth]{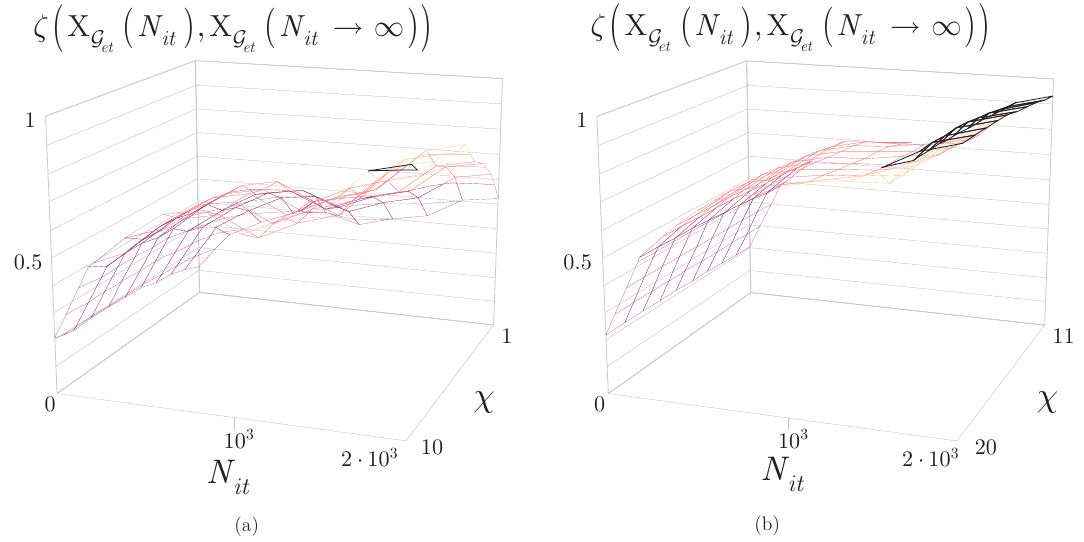}
\caption{The ratio $\zeta \left({\rm X} _{{\rm {\mathcal G}}_{et} } \left(N_{it} \right),{\rm X} _{{\rm {\mathcal G}}_{et} } \left(N_{it} \to \infty \right)\right)$ in function of $N_{it} $ and $\chi $. (a): The values of $\zeta \left({\rm X} _{{\rm {\mathcal G}}_{et} } \left(N_{it} \right),{\rm X} _{{\rm {\mathcal G}}_{et} } \left(N_{it} \to \infty \right)\right)$ for $0<N_{it} \le 2\cdot 10^{3} $ and $1\le \chi \le 10$. (b): The values of $\zeta \left({\rm X} _{{\rm {\mathcal G}}_{et} } \left(N_{it} \right),{\rm X} _{{\rm {\mathcal G}}_{et} } \left(N_{it} \to \infty \right)\right)$ for $0<N_{it} \le 2\cdot 10^{3} $ and $11\le \chi \le 20$.} 
 \label{fig_a2}
 \end{center}
\end{figure*}
\end{center}

The performance of the quantum layer optimization method is therefore approachable via the distance function \eqref{ZEqnNum594299} and ratio \eqref{ZEqnNum872057}. The analysis revealed that at moderate $N_{it} $ values, the precision of the optimization method can be arbitrary high via the selection of the $\chi $ control parameter. It also has been concluded, that a high ratio of the solutions at a finite and moderate $N_{it} $, are identical to the solutions at the limit case of $N_{it} \to \infty $.

\subsection{Classical Layer Optimization}

Let $\phi _{s} \left(i\right)$ be the step-size function defined for an $i$-th state as
\begin{equation} \label{ZEqnNum765012} 
\phi _{s} \left(i\right)=\phi _{\max } -\Delta \left(\phi \right)\cdot \exp \left({\tfrac{-f\left(i,j,k,l\right)}{j}} \right),                                            
\end{equation} 
where  
\begin{equation} \label{7)} 
\Delta \left(\phi \right)=\left(\phi _{\max } -\phi _{\min } \right),                                                            
\end{equation} 
where $\phi _{\min } $ and $\phi _{\max } $ are some lower and upper bounds, $\phi _{s} \left(i\right)\in \left[\phi _{\min } ,\phi _{\max } \right]$, while $f\left(i,j,k,l\right)$ is defined as
\begin{equation} \label{ZEqnNum929340} 
f\left(i,j,k,l\right)=\omega {\tfrac{J^{*} }{J_{P} }} ,                                                                
\end{equation} 
where $\omega $ is a constant (restriction factor \cite{ref18, ref19, ref20}), $J^{*} $ is the minimal cost function (for $J$, see \eqref{ZEqnNum495559}) at a particular parameter setting $\left(i,j,k,l\right)$, $J_{P} $ is the minimal cost function associated to the population.

Then, at an iteration number $j$ let $\kappa \left(j\right)$ be a ratio as
\begin{equation} \label{9)} 
\kappa \left(j\right)={\tfrac{f\left(i,j,k,l\right)}{j}} . 
\end{equation} 
The step-size \eqref{ZEqnNum765012} in function of $\phi _{\min } $ and $\kappa \left(j\right)$ is depicted in \fref{fig_a3}. 

 \begin{center}
\begin{figure*}[!h]
\begin{center}
\includegraphics[angle = 0,width=0.9\linewidth]{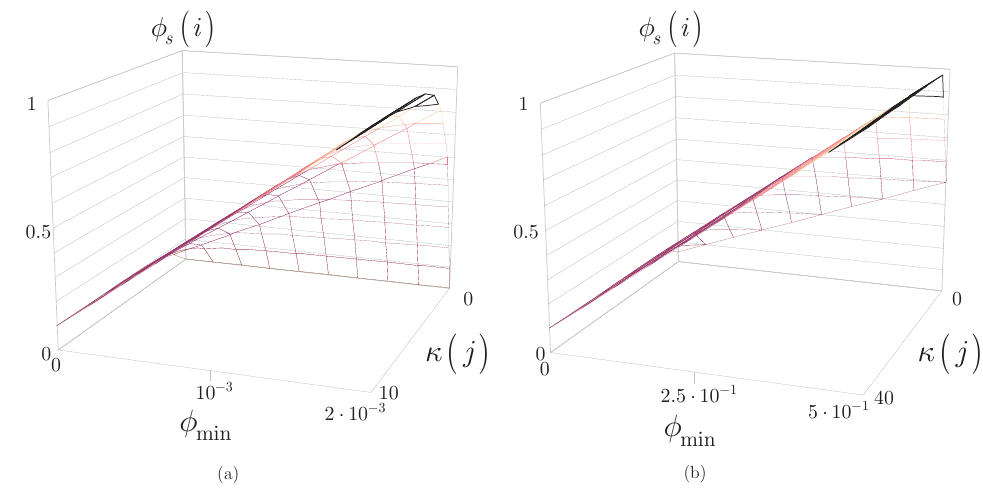}
\caption{(a): The $\phi _{s} \left(i\right)$ step-size function for $\phi _{\min } \in \left[0,2\cdot 10^{-3} \right]$, $\phi _{\max } =500\phi _{\min } $ and $\kappa \left(j\right)\in \left[0,10\right]$. (b): The $\phi _{s} \left(i\right)$ step-size function for $\phi _{\min } \in \left[0,5\cdot 10^{-1} \right]$, $\phi _{\max } =2\phi _{\min } $ and $\kappa \left(j\right)\in \left[0,40\right]$.} 
 \label{fig_a3}
 \end{center}
\end{figure*}
\end{center}

From \eqref{ZEqnNum765012} and \eqref{ZEqnNum929340}, the optimal cost function $J^{*} $ at a particular setting of $\left(i,j,k,l\right)$ is yielded as
\begin{equation} \label{ZEqnNum677599} 
J^{*} =-{\tfrac{\ln \left({\tfrac{\phi _{\max } -\phi _{s} \left(i\right)}{\Delta \left(\phi \right)}} \right)j}{\omega }} J_{P} .                                                           
\end{equation} 
Then, let 
\begin{equation} \label{ZEqnNum633781} 
\kappa \left(j\right)=x,                                                                        
\end{equation} 
from which the iteration number can be rewritten as
\begin{equation} \label{ZEqnNum484416} 
j={\tfrac{f\left(i,j,k,l\right)}{x}} ,                                                                        
\end{equation} 
and also fix $f\left(i,j,k,l\right)$ as
\begin{equation} \label{ZEqnNum807141} 
f\left(i,j,k,l\right)=x^{2} .                                                                       
\end{equation} 
Then, from \eqref{ZEqnNum633781} and \eqref{ZEqnNum807141}, $J^{*} $ can be evaluated at a given $x$ in function $J_{P} $ and $\omega $, as
\begin{equation} \label{ZEqnNum238542} 
J^{*} ={\tfrac{x^{2} J_{P} }{\omega }} .                                                                               
\end{equation} 
From \eqref{ZEqnNum484416} and \eqref{ZEqnNum807141}, the function in \eqref{ZEqnNum677599} can be rewritten as
\begin{equation} \label{15)} 
\begin{split}
   {{J}^{*}}&=-\tfrac{\ln \left( \tfrac{{{\phi }_{\max }}-{{\phi }_{s}}\left( i \right)}{\Delta \left( \phi  \right)} \right)f\left( i,j,k,l \right)}{\omega x}{{J}_{P}} \\ 
 & =-\tfrac{\ln \left( \tfrac{{{\phi }_{\max }}-{{\phi }_{s}}\left( i \right)}{\Delta \left( \phi  \right)} \right)x}{\omega }{{J}_{P}},  
\end{split}
\end{equation} 
therefore the cost function $J^{*} $ at a particular \eqref{ZEqnNum633781} and \eqref{ZEqnNum807141} can also be evaluated in function of the step size \eqref{ZEqnNum765012}. 

The $J^{*} $ cost function values associated to the $\phi _{s} \left(i\right)$ step size function values of \fref{fig_a3} at $J_{P} =1$ and $J_{P} =1$ and $\omega =100$ are depicted in \fref{fig_a4}. 

 \begin{center}
\begin{figure*}[!h]
\begin{center}
\includegraphics[angle = 0,width=0.9\linewidth]{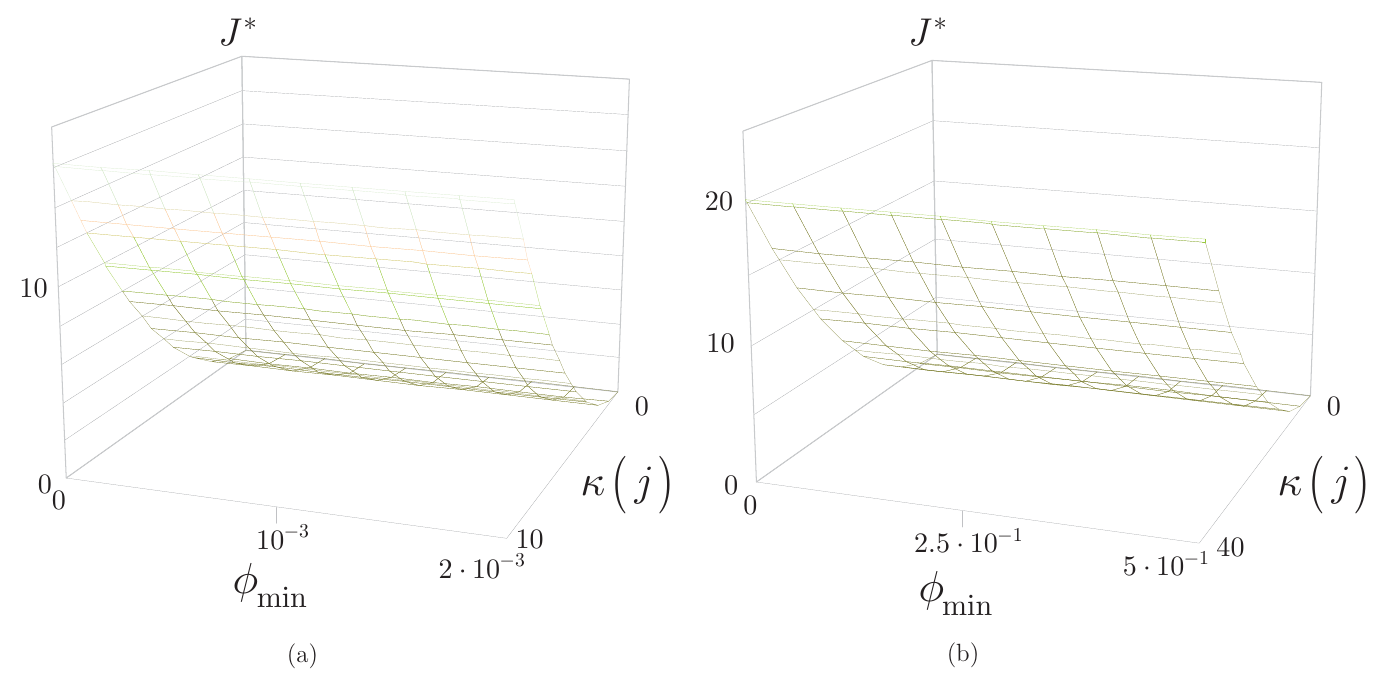}
\caption{ The $J^{*} $ cost function values associated to the $\phi _{s} \left(i\right)$ step size function value, (a): $J_{P} =1$ and $\omega =1$, (b): $J_{P} =1$ and $\omega =100$.}
 \label{fig_a4}
 \end{center}
\end{figure*}
\end{center}

The analysis is therefore revealed that for any $\phi _{s} \left(i\right)$, the cost function $J^{*} $ increases with the $\kappa \left(j\right)$ ratio, and the values of $J^{*} $ are tunable via the control parameter $\omega $ for arbitrary $J_{P} $ and $\kappa \left(j\right)$ values. The proposed classical layer optimization model is therefore flexible, and allows a dynamics adaption for diverse environmental settings.
As future work, our aim is to provide a transmission analysis and comparisons with other schemes.

\section{Conclusions}
\label{sec6}
Here we conceived a multilayer optimization method for the quantum Internet. Multilayer optimization defines separate procedures for the optimization of the quantum layer and the classical layer. Quantum-layer optimization defines a multi-objective function by minimizing the total usage of quantum memories in the quantum nodes, maximizing the entanglement throughput over all entangled links, and reducing the number of entangled links between the arbitrary source and target quantum nodes. We defined the structure of the quantum memory utilization graph and entanglement throughput tree. The classical-layer optimization utilizes the fundaments of swarm intelligence for the minimization of the cost function. Since the proposed multilayer optimization method has no physical layer requirements, it can serve as a useful tool for quantum network communications and future quantum Internet. 

\section*{Acknowledgements}
This work was partially supported by the National Research Development and Innovation Office of Hungary (Project No. 2017-1.2.1-NKP-2017-00001), by the Hungarian Scientific Research Fund - OTKA K-112125 and in part by the BME Artificial Intelligence FIKP grant of EMMI (BME FIKP-MI/SC).

\bibliography{report}

\begin{thebibliography}{10}
\bibitem{ref1} Van Meter, R. \textit{Quantum Networking}. ISBN 1118648927, 9781118648926, John Wiley and Sons Ltd (2014).

\bibitem{ref2} Gyongyosi, L., Imre, S. and Nguyen, H. V. A Survey on Quantum Channel Capacities, \textit{IEEE Communications Surveys and Tutorials} \textbf{99}, 1, doi: 10.1109/COMST.2017.2786748 (2018).

\bibitem{ref3} Pirandola, S. Capacities of repeater-assisted quantum communications, \textit{arXiv:1601.00966} (2016). 

\bibitem{ref4} Gyongyosi, L., Imre, S. Entanglement-Gradient Routing for Quantum Networks, \textit{Sci. Rep.}, Nature, DOI: OI:10.1038/s41598-017-14394-w (2017).

\bibitem{ref5} Lloyd, S., Shapiro, J. H., Wong, F. N. C., Kumar, P., Shahriar, S. M. and Yuen, H. P. Infrastructure for the quantum Internet. \textit{ACM SIGCOMM Computer Communication Review}, 34, 9--20 (2004).

\bibitem{ref6} Imre, S., Gyongyosi, L. \textit{Advanced Quantum Communications - An Engineering Approach}. New Jersey, Wiley-IEEE Press (2013).

\bibitem{ref7} Gyongyosi, L., Imre, S.  Entanglement Availability Differentiation Service for the Quantum Internet, \textit{Sci. Rep.}, Nature, DOI:10.1038/s41598-018-28801-3 (2018). 

\bibitem{ref8} Jie, Y., Kamal, A. E. Multi-Objective Multicast Routing Optimization in Cognitive Radio Networks, \textit{IEEE Wireless Communications and Networking Conference} (IEEE WCNC) (2014).

\bibitem{ref9} Pinto, D., Baran, B. Solving multiobjective multicast routing problem with a new ant colony optimization approach. \textit{Proceedings of ACM, International IFIP/ACM Latin American Conference on Networking} (2005).

\bibitem{ref10} Wang, W. et al. Efficient interference-aware TDMA link scheduling for static wireless networks. \textit{Proceedings of ACM, International Conference on Mobile Computing and Networking} (2006). 

\bibitem{ref11} Rani, B. S., Kumar, C. A. A Comprehensive Review on Bacteria Foraging Optimization Technique, \textit{Multi-objective Swarm Intelligence}, Theoretical Advances and Applications, Studies in Computational Intelligence, Volume 592, Springer (2015).

\bibitem{ref12} Liu, Y., Passino, K. M. Biomimicry of Social Foraging Bacteria for Distributed Optimization: Models, Principles, and Emergent Behaviors, \textit{Journal of Optimization Theory and Applications} Vol. 115, No. 3, pp. 603--628 (2002).

\bibitem{ref13} Motevasel, M., Bazyari, S. Probabilistic Energy Management of micro-grids with respect to Economic and Environmental Criteria, \textit{Science Journal} (CSJ), Vol. 36, No: 3 Special Issue, ISSN: 1300-1949 (2015). 

\bibitem{ref14} Farooq, M., Di Caro, G. A. Routing protocols for next generation networks inspired by collective behaviors of insect societies: an overview. \textit{Swarm Intelligence, Natural Computing Series}, pages 101-160 (2008).

\bibitem{ref15} Di Caro, F. D. G., Gambardella, L. M. Swarm intelligence for routing in mobile ad hoc networks. \textit{IEEE Swarm Intelligence Symposium}, pages 76-83 (2005).

\bibitem{ref16} Saleem, G. A. D. C. M., Farooq, M. Swarm intelligence based routing protocol for wireless sensor networks: survey and future directions. \textit{Information Sciences}, 181(20):4597-4624 (2011).

\bibitem{ref17} Neumann, F., Witt, C. Bioinspired computation in combinatorial optimization algorithms and their computational complexity. \textit{Natural Computing Series} (2010).

\bibitem{ref18} Zhang, Y., Zhou, W. and Yi, J. A Novel Adaptive Chaotic Bacterial Foraging Optimization Algorithm, \textit{2016 International Conference on Computational Modeling, Simulation and Applied Mathematics}, (2016).

\bibitem{ref19} Niu, B., Fan, Y., Xiao, H. and Bing, X. Bacterial foraging based approaches to portfolio optimization with liquidity risk, \textit{Neruocomputing} 98 (2012) 90-100.

\bibitem{ref20} Li, M. S., Ji, T. Y., Tang, W. J., Wu, Q. H. and Saunders, J. R. Bacterial foraging algorithm with varying population, \textit{Biosystems} 100(3)(2010) 185-197.

\bibitem{ref21} Yuan, Z., Chen, Y., Zhao, B., Chen, S., Schmiedmayer, J. and Pan, J. W. \textit{Nature} 454, 1098-1101 (2008).

\bibitem{ref22} Kok, P., Munro, W. J., Nemoto, K., Ralph, T. C., Dowling, J. P. and Milburn, G. J. Linear optical quantum computing with photonic qubits, \textit{Rev. Mod. Phys}. 79 , 135-174 (2007)

\bibitem{ref23} Biamonte, J. et al. Quantum Machine Learning. \textit{Nature}, 549, 195-202 (2017). 

\bibitem{ref24} Lloyd, S., Mohseni, M. and Rebentrost, P. Quantum algorithms for supervised and unsupervised machine learning. \textit{arXiv:1307.0411} (2013).

\bibitem{ref25} Lloyd, S., Mohseni, M. and Rebentrost, P. Quantum principal component analysis. \textit{Nature Physics}, 10, 631 (2014).

\bibitem{ref26} Lloyd, S. Capacity of the noisy quantum channel. \textit{Physical Rev. A}, 55:1613--1622 (1997).

\bibitem{ref27} Gisin, N., Thew, R. Quantum Communication. \textit{Nature Photon}. 1, 165-171 (2007).

\bibitem{ref28} Shor, P. W. Scheme for reducing decoherence in quantum computer memory. \textit{Phys. Rev. A}, 52, R2493-R2496 (1995).

\bibitem{ref29} Chou, C., Laurat, J., Deng, H., Choi, K. S., de Riedmatten, H., Felinto, D. and Kimble, H. J. Functional quantum nodes for entanglement distribution over scalable quantum networks. \textit{Science}, 316(5829):1316--1320 (2007).

\bibitem{ref30} Sheng, Y. B., Zhou, L. Distributed secure quantum machine learning. \textit{Science Bulletin}, 62, 1025-2019 (2017). 

\bibitem{ref31} Kimble, H. J. The quantum Internet. \textit{Nature}, 453:1023--1030 (2008).

\bibitem{ref32} Muralidharan, S., Kim, J., Lutkenhaus, N., Lukin, M.D. and Jiang. L. Ultrafast and Fault-Tolerant Quantum Communication across Long Distances, \textit{Phys. Rev. Lett}. 112, 250501 (2014).

\bibitem{ref33} Lloyd, S. The Universe as Quantum Computer, \textit{A Computable Universe: Understanding and exploring Nature as computation}, H. Zenil ed., World Scientific, Singapore, \textit{arXiv:1312.4455v1} (2013).

\bibitem{ref34} Pirandola, S., Laurenza, R., Ottaviani, C. and Banchi, L. Fundamental limits of repeaterless quantum communications, \textit{Nature Communications}, 15043, doi:10.1038/ncomms15043 (2017).

\bibitem{ref35} Pirandola, S., Braunstein, S.L., Laurenza, R., Ottaviani, C., Cope, T.P.W., Spedalieri, G. and Banchi, L. Theory of channel simulation and bounds for private communication, \textit{Quantum Sci. Technol.} 3, 035009 (2018).

\bibitem{ref36} Laurenza, R. and Pirandola, S. General bounds for sender-receiver capacities in multipoint quantum communications, \textit{Phys. Rev. A} 96, 032318 (2017).

\bibitem{ref37} Vedral, V., Plenio, M.B., Rippin, M.A. and Knight, P.L. Quantifying Entanglement, \textit{Phys. Rev. Lett}. 78, 2275-2279 (1997). 

\bibitem{ref38} Vedral, V. The role of relative entropy in quantum information theory, \textit{Rev. Mod. Phys}. 74, 197–234 (2002).

\bibitem{ref39} Petz, D. \textit{Quantum Information Theory and Quantum Statistics}, Springer-Verlag, Heidelberg, Hiv: 6. (2008).

\bibitem{ref40} Bacsardi, L. On the Way to Quantum-Based Satellite Communication, \textit{IEEE Comm. Mag.} 51:(08) pp. 50-55. (2013). 

\bibitem{ref41} Gyongyosi, L. and Imre, S. Decentralized Base-Graph Routing for the Quantum Internet, \textit{Phys. Rev. A}, American Physical Society, DOI: 10.1103/PhysRevA.98.022310 (2018).

\bibitem{ref42} Caleffi, M. End-to-End Entanglement Rate: Toward a Quantum Route Metric, \textit{2017 IEEE Globecom}, DOI: 10.1109/GLOCOMW.2017.8269080, (2018). 

\bibitem{ref43} Van Meter, R., Satoh, T., Ladd, T.D., Munro, W.J. and Nemoto, K. Path selection for quantum repeater networks, \textit{Networking Science}, Volume 3, Issue 1–4, pp 82–95, (2013).

\bibitem{ref44} Caleffi, M. Optimal Routing for Quantum Networks, \textit{IEEE Access}, Vol 5, DOI: 10.1109/ACCESS.2017.2763325 (2017).

\bibitem{ref45} Caleffi, M., Cacciapuoti, A.S. and Bianchi, G. Quantum Internet: from Communication to Distributed Computing, \textit{aXiv:1805.04360} (2018).
 
\bibitem{ref46} Castelvecchi, D. The quantum internet has arrived, \textit{Nature}, News and Comment, https://www.nature.com/articles/d41586-018-01835-3, (2018).


\end{thebibliography}
\bibliographystyle{unsrt}

\newpage
\appendix
\setcounter{table}{0}
\setcounter{figure}{0}
\setcounter{equation}{0}
\renewcommand{\thetable}{\Alph{section}.\arabic{table}}
\renewcommand{\thefigure}{\Alph{section}.\arabic{figure}}
\renewcommand{\theequation}{\Alph{section}.\arabic{equation}}

\setlength{\arrayrulewidth}{0.1mm}
\setlength{\tabcolsep}{5pt}
\renewcommand{\arraystretch}{1.5}

\maketitle

\section{Appendix}

\subsection{Notations}
The notations of the manuscript are summarized in  \tref{tab1}.
\begin{center}
\begin{longtable}{||l|p{4.2in}||}
\caption{Summary of notations.}
\label{tab1}
\endfirsthead
\endhead
\hline
\textit{Notation} & \textit{Description} \\ \hline
$l$  & Level of entanglement.  \\ \hline 
$F$ & Fidelity of entanglement.  \\ \hline 
$N$ & Entangled quantum network, $N=\left(V,{\rm {\mathcal S}}\right)$, where $V$ is a set of nodes, ${\rm {\mathcal S}}$ is a set of entangled links. \\ \hline 
${\rm L}_{l} $ & An $l$-level entangled link. For an ${\rm L}_{l} $ link, the hop-distance is $2^{l-1} $. \\ \hline 
$d\left(x,y\right)_{{\rm L}_{l} } $ & Hop-distance of an $l$-level entangled link between nodes $x$ and $y$.  \\ \hline 
$E_{{\rm L}_{l} } \left(x,y\right)$ & Entangled link $E_{{\rm L}_{l} } \left(x,y\right)$ between nodes $x$ and $y$. \\ \hline 
$\lambda _{E_{{\rm L}_{l} } \left(x,y\right)} $ & Initial entanglement utility of link $E_{{\rm L}_{l} } \left(x,y\right)$. \\ \hline 
$\lambda '_{E_{{\rm L}_{l} } \left(x,y\right)} $ & Updated entanglement utility of link $E_{{\rm L}_{l} } \left(x,y\right)$. \\ \hline 
$B_{F} \left(E_{{\rm L}_{l} } \left(x,y\right)\right)$ & Entanglement throughput of a given ${\rm L}_{l} $-level entangled link $E_{{\rm L}_{l} } \left(x,y\right)$ between nodes $\left(x,y\right)$. \\ \hline 
$R_{i} $ & An $i$--th quantum repeater node. \\ \hline 
$S$ & A quantum switcher node, switches between entangled connects using its local quantum memory, and applies entanglement swapping.  \\ \hline 
${\rm {\mathcal G}}_{m} $ & Quantum memory utilization graph, directed graph mapped from the network model with abstracted nodes and links. \\ \hline 
${\rm {\mathcal G}}_{et} $ & Entanglement throughput tree is derived from a ${\rm {\mathcal G}}_{m} $ quantum memory utilization graph. \\ \hline 
$ID$ & Identifier of a node in ${\rm {\mathcal G}}_{et} $, $ID=\left\{A,B,\ldots \right\}$. \\ \hline 
${\rm S}_{I} $ & Set of unvisited neighbor nodes of a particular node $I$. \\ \hline 
$\Omega \left(I,J\right)$ & Cost function between nodes $\left(I,J\right)$. \\ \hline 
$C\left(E_{{\rm L}_{l} } \left(I,J\right)\right)$ & Cost of entangled link $E_{{\rm L}_{l} } \left(I,J\right)$. \\ \hline 
$\zeta _{J} $ & Cost of quantum storage in node $J$. \\ \hline 
$\Pr \left(I,J\right)$ & Probability that from node $I$ a node $J$ is selected. \\ \hline 
$\chi $, $\delta $ & Weighting coefficients in $\Pr \left(I,J\right)$. \\ \hline 
${\rm {\mathcal M}}$ & Method of building an ${\rm {\mathcal G}}_{et} $ entanglement throughput tree. \\ \hline 
${\rm S}'$ & Set of already reached destination nodes. \\ \hline 
${\rm {\mathcal I}}$ & Set of initial nodes. \\ \hline 
${\rm {\mathcal F}}_{I} $ & Set of feasible neighboring nodes to node $I$. \\ \hline 
${\rm {\mathcal D}}$ & Set of destination nodes.    \\ \hline 
$\alpha _{{\rm {\mathcal G}}_{et} } $ & Entanglement assignment cycle, an optimal assignment (scheduling) of stored entanglement. \\ \hline 
$t_{s}^{*} \left({\rm {\mathcal G}}_{et} \right)$ & Minimal overall storage time at a given ${\rm {\mathcal G}}_{et} $. \\ \hline 
${\rm {\mathcal C}}_{{\rm {\mathcal G}}_{et} } $ & Conflict graph of ${\rm {\mathcal G}}_{et} $. In the ${\rm {\mathcal C}}_{{\rm {\mathcal G}}_{et} } $ graph, each vertex corresponds to a directed link of ${\rm {\mathcal G}}_{et} $ (an entangled connection). There is an edge between two vertices of ${\rm {\mathcal C}}_{{\rm {\mathcal G}}_{et} } $, if only the vertices (entangled connections) has a conflict. \\ \hline 
$\tau _{n,t}$ & Indicator variable, $\tau _{n,t} \in \left\{0,1\right\}$, defined as\newline $\tau _{n,t} =\left\{\begin{array}{l} {1,{\rm \; if\; }n{\rm \; is\; associated\; at\; time\; }t{\rm \; \; \; \; }} \\ {0,{\rm \; otherwise.\; \; \; \; \; \; \; \; \; \; \; \; \; \; \; \; \; \; \; \; \; \; \; \; \; }} \end{array}\right. $ \\ \hline 
$\wedge \left(n\right)$ & Set of entangled connects $n'$ that are scheduled in the same time unit $t$, but the physical link can transmit only $n$ or $n'$. \\ \hline 
$w\left(n\right)$ & Weight of an entangled link, defined as \newline $w\left(n\right)=\left\{\begin{array}{l} {1,{\rm \; \; \; \; \; \; \; \; if\; }F_{i} =F_{\max } } \\ {\left\lceil \frac{F_{\max } }{F_{i} } \right\rceil ,{\rm \; if\; }F_{i} <F_{\max } } \end{array}\right. ,$                                        \newline where $F_{i} $ is the fidelity of entangled connection $i$, $F_{\max } $ is the largest fidelity. \\ \hline 
${\rm {\mathcal W}}\left({\rm {\mathcal C}}_{{\rm {\mathcal G}}_{et} } \right)$ & Weighted coloring of conflict graph ${\rm {\mathcal C}}_{{\rm {\mathcal G}}_{et} } $. \\ \hline 
$\Delta \left({\rm {\mathcal W}}\left({\rm {\mathcal C}}_{{\rm {\mathcal G}}_{et} } \right)\right)$ & Time intervals between each time unit of a given cycle. \\ \hline 
${\rm {\mathcal G}}_{et}^{*} $ & Optimal entanglement throughput tree. \\ \hline 
$t_{s} \left({\rm {\mathcal G}}_{et}^{*} \right)$ & Overall storage time at an optimal ${\rm {\mathcal G}}_{et}^{*} $. \\ \hline 
$B_{F} \left({\rm {\mathcal G}}_{et}^{*} \right)$ & Entanglement throughput at an optimal ${\rm {\mathcal G}}_{et}^{*} $. \\ \hline 
$\left|{\rm {\mathcal P}}\left({\rm {\mathcal G}}_{et}^{*} \right)\right|$ & Number of entangled links at an optimal ${\rm {\mathcal G}}_{et}^{*} $. \\ \hline 
${\rm {\mathcal S}}_{{\rm {\mathcal G}}_{et} }^{*} $ & Set of optimal ${\rm {\mathcal G}}_{et}^{*} $ entanglement throughput trees. \\ \hline 
${\rm X} $ & Solution set ${\rm X} $ with decision variables \newline ${\rm X} =\left\{x_{1} ,\ldots ,x_{n} \right\}$,\newline where $n$ is the number of all links in a given quantum memory utilization graph ${\rm {\mathcal G}}_{m} $, and $x_{i} \in \left\{0,1\right\}$ is defined as\newline $x_{i} =\left\{\begin{array}{l} {1,{\rm \; if\; link\; }i{\rm \; of\; {\mathcal G}}_{m} {\rm \; is\; selected\; by\; method\; {\mathcal M}}} \\ {0,{\rm \; otherwise.\; \; \; \; \; \; \; \; \; \; \; \; \; \; \; \; \; \; \; \; \; \; \; \; \; \; \; \; \; \; \; \; \; \; \; \; \; }} \end{array}\right. $ \\ \hline 
${\rm X} _{{\rm {\mathcal G}}_{et} ,j} \angle {\rm X} _{{\rm {\mathcal G}}_{et} ,i} $ & Set ${\rm X} _{{\rm {\mathcal G}}_{et} ,i} $ dominates ${\rm X} _{{\rm {\mathcal G}}_{et} ,j} $. \\ \hline 
$\kappa $ & Set that contains the best non-dominated solutions that have been found at a particular iteration. \\ \hline 
$f_{t_{s}^{*} ,B_{F}^{*} ,\left|{\rm {\mathcal P}}^{*} \right|} \left(\Theta _{i} \right)$ & Cost function of classical-layer optimization, where $\Theta _{i} \in {\mathbb{R} }^{p} $ is a $p$-dimensional real vector of an $i$-th system state of the quantum network.  \\ \hline 
$\Theta _{i} \in {\mathbb{R} }^{p} $ & A $p$-dimensional real vector of an $i$-th system state of the quantum network.   \\ \hline 
$\Theta _{i} \left(j,k,l\right)$ & An $i$-th system state, where $j$ is the index of a desired optimal system state, $k$ is the index of an optimal system state reproduction step, $l$ is the index of a non-optimal system state event. \\ \hline 
$T$ & Total network state, $T\left(j,k,l\right)=\left\{\Theta _{i} \left. \left(j,k,l\right)\right|i=1,\ldots ,S\right\}$,                            at a set of $S$ sub-states $\left\{\Theta _{1} ,\ldots ,\Theta _{S} \right\}$. \\ \hline 
$c\left(i\right)$ & Number of random system states. \\ \hline 
$u\left(j\right)$ & A unit cost of system change. \\ \hline 
$C_{N} $ & Total cost of classical communication. \\ \hline 
${\rm A} $ & Distribution-entity of a current system state. \\ \hline 
$R_{{\rm A} } $ & Information transmission rate of ${\rm A} $. \\ \hline 
$\nu $ & Distribution-entity of a system state. \\ \hline 
$R_{\nu } $ & Information transmission rate of $\nu $. \\ \hline 
$\Theta ^{\left(m\right)} $ & The $m$-th element of a current network state vector $\Theta $. \\ \hline 
$\Theta _{i}^{\left(m\right)} $ & The $m$-th element of $\Theta _{i} $. \\ \hline 
$C_{e} $ & An environment-dependent cost function. \\ \hline 
$M$ & Tuning parameter. \\ \hline 
$F_{cost}^{i} $ & Cost function at a given $\Theta _{i} \left(j,k,l\right)$. \\ \hline 
$N_{t_{s}^{*} } $, $N_{B_{F}^{*} } $, $N_{\left|{\rm {\mathcal P}}^{*} \right|} $ & The number of nodes that require the determination of optimal $t_{s}^{*} $, $B_{F}^{*} $ and $\left|{\rm {\mathcal P}}^{*} \right|$. \\ \hline 
$S_{t_{s}^{*} } \left(t\right)$, $S_{B_{F}^{*} } \left(t\right)$, $S_{\left|{\rm {\mathcal P}}^{*} \right|} \left(t\right)$ & The number of classical steps required to find $t_{s}^{*} $, $B_{F}^{*} $ and $\left|{\rm {\mathcal P}}^{*} \right|$ at a particular network time $t$, $t=1,\ldots ,T$. \\ \hline 
$J$ & Objective function. \\ \hline 
$c_{N_{t_{s}^{*} } }^{L} \left(t\right)$, $c_{N_{B_{F}^{*} } }^{L} \left(t\right)$, $c_{N_{\left|{\rm {\mathcal P}}^{*} \right|} }^{L} \left(t\right)$ & Link cost of classical link $L$ used for the determination of $t_{s}^{*} $, $B_{F}^{*} $ and $\left|{\rm {\mathcal P}}^{*} \right|$ at a particular time $t$. \\ \hline 
$\Theta _{M} \left(j,k,l\right)$ & Merged system state for the optimization of the classical-layer. \\ \hline 
$\Phi $ & Merging factor, $\Phi \in \left[0,1\right]$. \\ \hline 
$u$ & Uniform random number. \\ \hline 
$x_{a} ,x_{b} $,$x_{c} $ & Random numbers, $x_{a} ,x_{b} ,x_{c} \in \left[0,1\right]$. \\ \hline 
$O_{J} =f\left(\psi _{in} \right)$ & Output variable, where $\psi _{in} $ is the set of input variables, and $f\left(\cdot \right)$ is a function that transfers the uncertainty from the independent input random variables $\psi _{in} $ to the output variable $O_{J} $ \\ \hline 
$q$ & Set of certain variables. \\ \hline 
$w_{i} $ & An input variable under certainty with probability function $\delta _{f_{w_{i} } } $. \\ \hline 
$pc\left(w_{i} \right)$ & Probability concentration of $w_{i} $. \\ \hline 
$\zeta _{i,g} $ & Weighting factor. \\ \hline 
$O_{J}^{\left(i,g\right)} $ & Output variable $O_{J} $ at a given $\left(i,g\right)$. \\ \hline 
$\mu _{w_{i} } $ & Mean of $w_{i} $. \\ \hline 
$w_{i,1} $,$w_{i,2} $ & Poth locations of $w_{i} $. \\ \hline 
\end{longtable}
\end{center}
\end{document}